\documentclass[journal]{IEEEtran}
\ifCLASSINFOpdf
\else
\fi
\usepackage{grffile}
\usepackage{amsmath,amssymb,amsbsy}
\usepackage{hyperref}
\usepackage{amsthm} 
\usepackage{color}
\usepackage{comment}
\usepackage{lipsum}
\usepackage{graphicx}%
\usepackage{bm}
\usepackage{subfigure}
\usepackage[normalem]{ulem}
\usepackage{enumerate}
\usepackage{epstopdf}
\usepackage{thmtools, thm-restate}
\newtheorem{theorem}{Theorem}[section]
\newtheorem{lemma}[theorem]{Lemma}
\newtheorem{definition}[theorem]{Definition}	
\newtheorem{proposition}[theorem]{Proposition}
\newtheorem{corollary}[theorem]{Corollary}	
	
\newtheorem{problem}[theorem]{Problem}

\newcommand{\essup}{${\rm ess sup}$}
\allowdisplaybreaks

\begin{document}
%
\title{Bilinear Controllability of a Class of Advection-Diffusion-Reaction Systems}
%
%
%

\author{Karthik Elamvazhuthi, Hendrik Kuiper, Matthias Kawski, and Spring Berman

\thanks{This work was supported by National Science Foundation (NSF)
award no. CMMI-1436960.}
\thanks{K. Elamvazhuthi and S. Berman are with the School
for Engineering of Matter, Transport and Energy, Arizona State
University, Tempe, AZ 85287, USA (e-mail: karthikevaz@asu.edu,
spring.berman@asu.edu).}
\thanks{H. Kuiper and M. Kawski are with the School of Mathematical and Statistical Sciences, Arizona State
University, Tempe, AZ 85287, USA (e-mail: kuiper@asu.edu,
kawski@asu.edu).}}

\maketitle

\begin{abstract}
In this paper, we investigate the exact controllability properties of an advection-diffusion equation on a bounded domain, using time- and space-dependent velocity fields as the control parameters. This partial differential equation (PDE) is the Kolmogorov forward equation for a reflected diffusion process that models the spatiotemporal evolution of a swarm of agents. We prove that if a target probability density has bounded first-order weak derivatives and is uniformly bounded from below by a positive constant, then it can be reached in finite time using control inputs that are bounded in space and time. We then extend this controllability result to a class of  advection-diffusion-reaction PDEs that corresponds to a hybrid-switching diffusion process (HSDP), in which case the reaction parameters are additionally incorporated as the control inputs.
 For the HSDP, we first constructively prove controllability of the associated continuous-time Markov chain (CTMC) system, in which the state space is finite. Then we show that our controllability results for the advection-diffusion equation and the CTMC can be combined to establish controllability of the forward equation of the HSDP. Lastly, we provide constructive solutions to the problem of asymptotically stabilizing an HSDP to a target non-negative stationary distribution using time-independent state feedback laws, which correspond to spatially-dependent coefficients of the associated system of PDEs.
\end{abstract}


%
\IEEEpeerreviewmaketitle

\section{Introduction}

In recent years, there has been an extensive amount of work on control of large-scale multi-agent systems. A particular instance of these control problems is when the agents represent autonomous robots that must collectively achieve some desired global objective. This has motivated various investigations on the appropriate technical framework for modeling and control of such systems, called {\it robotic swarms}. One such framework involves modeling swarms using probability densities and controlling them through mean-field models \cite{brockett2012notes,martinoli2004modeling,
milutinovic2006modeling}. Applications of these models to the control of robotic swarms has led to new questions from a control-theoretical point of view. One fundamental question is that of controllability: {\it given a mean-field model with a control parameter, what is the class of reachable probability densities?} Stated as such, the problem is not new since the use of mean-field models in the natural sciences and engineering is classical. 
However, the choice of control parameters in the context of swarm control has led to some non-classical controllability problems. 


As one of the main contributions of this paper, we consider such a controllability problem for a robotic swarm  that is described by a mean-field model in the form of an {\it advection-diffusion}  partial differential equation (PDE). Similar controllability problems have been addressed previously in the literature. For example, motivated by problems arising from quantum physics, Blaquiere \cite{blaquiere1992controllability} used techniques from stochastic control to study a controllability problem in which a stochastic process evolves on a $n$-dimensional Euclidean space $\mathbb{R}^n$. A similar controllability result was proved in \cite{dai1991stochastic}. In \cite{porretta2014planning}, Porretta addressed a controllability problem for a Fokker-Planck equation evolving on the $n$-dimensional torus, along with an associated {\it mean-field game} problem \cite{bensoussan2013mean}. This work applied observability inequalities that are typically used in PDE controllability problems. The results in \cite{blaquiere1992controllability,dai1991stochastic} were extended to a more general setting in which the stochastic process is a linear control system perturbed by a diffusion process \cite{chen2017optimal}. Controllability problems for systems with a similar structure have also been considered in work on {\it multiplicative control} of PDEs \cite{khapalov2010controllability}.

In contrast to these works, in this paper, the stochastic process that models the agents' motion is confined to a bounded subset of a Euclidean space. Boundedness of the domain is a common constraint in many problems in swarm robotics, e.g. in \cite{elamvazhuthi2018optaut,milutinovic2006modeling}, where optimal control techniques were used to optimize swarm behavior. Additionally, the results in previous controllability studies were proven with control parameters that are square-integrable. However, in bilinear optimal control of PDEs associated with stochastic processes, the boundedness of the vector fields is a common requirement \cite{elamvazhuthi2018optaut,finotti2012optimal,fleig2016optimal}. Toward this end, we establish controllability with control inputs that are (essentially) bounded in space and time. 

Another contribution of this paper is our analysis of a controllability problem for the forward equation of a class of hybrid switching diffusion processes (HSDPs) \cite{yin2010hybrid}. These processes can be used as models for robots that switch between multiple behavioral states, e.g. \cite{milutinovic2006modeling,elamvazhuthi2018optaut}. Our result is based on a controllability result for the forward equation of a related class of continuous-time Markov chains (CTMCs). A nontrivial issue in the problem of controlling the forward equation of CTMCs is the fact that the control parameters, which correspond to the transition rates of the Markov chain, are constrained to be positive. Hence, classical results on controllability of nonlinear control systems governed by ordinary differential equations (ODEs) do not apply. In spite of this issue, we prove controllability of these forward equations using piecewise constant control inputs. This controllability property can be attributed to the strong connectivity of the associated graphs.

As a final contribution of this paper, we consider the problem of stabilizing HSDPs to desired stationary distributions using time-independent and spatially-dependent controls or state feedback laws. 
A similar problem was considered in \cite{mesquita2012jump} for general controllable systems on unbounded domains with a single discrete state. 
 
The results presented in this paper are partially extensions of our prior work in \cite{elamvaz2016pde}, where controllability was proved for the case in which a swarm modeled using a Fokker-Planck equation 
evolves on a one-dimensional bounded domain. Here, we generalize this result to multi-dimensional domains with sufficiently smooth boundaries that satisfy the {\it chain condition} (see Definition \ref{chain}). Additionally, the requirement in \cite{elamvaz2016pde} for target densities to have second-order partial derivatives that are essentially bounded in space is relaxed to the requirement that only first-order partial derivatives must be bounded. The controllability result for the forward equation of CTMCs was first summarized in our prior work \cite{elamvaz2016lin} without proof. A complete proof of this result is given in this paper. The controllability and stabilization results presented in this paper for the case of HSDPs are new. 

The paper is organized as follows. In Section \ref{secII}, we establish notation and provide some definitions that are used throughout the paper. In Section \ref{prborm}, we formulate the main problems that are addressed in the paper. In Section \ref{sec:secctrban}, we present a detailed analysis of the controllability properties of the systems defined in Section \ref{prborm}. In Section \ref{stabsec}, we consider the problem of stabilizing target equilibrium densities of the system of PDEs using time-independent and spatially-dependent coefficients of the PDEs. Finally, we discuss conclusions in Section \ref{sec:conc}.

\section{NOTATION}
\label{secII}
We denote the $n$-dimensional Euclidean space by $\mathbb{R}^n$. $\mathbb{R}^{n \times m}$ refers to the space of $n \times m$ matrices, and $\mathbb{R}_+$ refers to the set of non-negative real numbers. Given a vector $\mathbf{x}  \in \mathbb{R}^n$, $x_i$ denotes the $i^{th}$ coordinate value of $\mathbf{x}$. For a matrix $\mathbf{A} \in \mathbb{R}^{n \times m}$, $A^{ij} $ refers to the element in the $i^{th}$ row and $j^{th}$ column of $\mathbf{A}$. For a subset $B \subset \mathbb{R}^M$, ${\rm int}(B)$ refers to the interior of the set $B$. $\mathbb{C}$, $\mathbb{C}_-$, and $\bar{\mathbb{C}}_-$ denote the set of complex numbers, the set of complex numbers with negative real parts, and the set of complex numbers with non-positive real parts, respectively. $\mathbb{Z}_+$ refers to the set of positive integers.
 We denote by $\Omega$ an open, bounded, and connected subset of a Euclidean domain $\mathbb{R}^n$. The  boundary of $\Omega$ is denoted by $\partial \Omega$.
 
\begin{definition}
We will say that $\Omega$ is a \textbf{$C^{1,1}$ domain} if each point $\mathbf{x} \in \partial \Omega$ has a neighborhood $\mathcal{N}$ such that $\Omega \cap \mathcal{N}$ is represented by the inequality $x_n < \gamma(x_1,...,x_{n-1})$ in some Cartesian coordinate system for some function $\gamma : \mathbb{R}^{n-1} \rightarrow \mathbb{R}$ that is at least once differentiable and has derivatives of order $1$ that are Lipschitz continuous. 
\end{definition}

For each $1\leq p< \infty$, we define $L^p(\Omega)$ as the Banach space of complex-valued measurable functions over the  set $\Omega$ whose absolute value raised to $p^{{\rm th}}$ power has finite integral. We define $L^{\infty}(\Omega)$ as the space of essentially bounded measurable functions on $\Omega$. The space $L^{\infty}(\Omega)$ is equipped with the norm
$
\|z\|_{\infty} =~ \essup_{\mathbf{x} \in \Omega} |z(\mathbf{x})|, 
$
where $ \essup_{\mathbf{x} \in \Omega}(\cdot)$ denotes the {\it essential supremum} attained by its argument over the domain $\Omega$. The space $L^2(\Omega)$ is a Hilbert space when equipped with the standard inner product, $\langle \cdot , \cdot \rangle_{2} :  L^2(\Omega) \times L^2(\Omega) \rightarrow \mathbb{C}$, given by 
$
\langle f,g\rangle_{2} = \int_{\Omega} f(\mathbf{x})\bar{g}(\mathbf{x})d\mathbf{x}
$
for each $f,g \in L^2(\Omega)$, where $\bar{g}$ is the complex conjugate of the function $g$. The norm $\|\cdot\|_{2}$ on the space $L^2(\Omega)$ is defined as
$
\|f\|_{2} = \langle f, f\rangle^{1/2}_{2}
$
for each $f \in L^2(\Omega)$. For a function $f \in L^2(\Omega)$ and a given constant $c$, we write $f \geq c$ to imply that $f$ is real-valued and $f(\mathbf{x})\geq c$ for almost every (a.e.) $\mathbf{x} \in \Omega$.

Let $f_{x_i}$ denote the first-order (weak) partial derivative of the function $f$ with respect to the coordinate $x_i$. Similarly,  $f_{x_i x_i}$ will denote the second-order partial derivative of the function $f$ with respect to the coordinate $x_i$. We define the Sobolev space $H^1(\Omega) = \big \lbrace f \in L^2(\Omega): f_{x_i} \in L^2(\Omega) \text{ for } 1 \leq i \leq N \big \rbrace$. We equip this space with the usual Sobolev norm $\|\cdot\|_{H^1}$, given by
$\|f\|_{H^1} = \Big( \|f\|^2_{2} + \sum_{i=1}^n\| f_{x_i}\|^2_{2}\Big)^{1/2} \nonumber
$
for each $f \in H^1(\Omega)$. The weak gradient of a function $f \in H^1(\Omega)$ will be denoted by $\nabla f = [f_{x_1}~...~f_{x_n}]^T$. 

\begin{definition}
\label{extn}
We will call $\Omega$ an \textbf{extension domain} if there exists a linear bounded operator $E : H^1(\Omega) \rightarrow H^1(\mathbb{R}^n)$ such that $(Ef)(\mathbf{x}) = f(\mathbf{x})$ for a.e. $\mathbf{x} \in \Omega$.
\end{definition}
An example of an extension domain is a domain with Lipschitz boundary \cite{agranovich2015sobolev}[Theorem 10.4.1]. Unless otherwise stated, the \textbf{default assumption} in this paper will be that \textbf{$\Omega$ is an extension domain}. The exponential stability results will only require this default assumption. However, to prove the controllability result, we will need the stronger assumption that the domain $\Omega$ is $C^{1,1}$ or convex. An additional assumption about the domain $\Omega$ will be needed to prove the controllability result, which motivates the following definition.

\begin{definition}
\label{chain}
The domain $\Omega$ will be said to satisfy the \textbf{chain condition} if there exists a constant $C>0$ such that for every $\mathbf{x},\bar{\mathbf{x}}$ $\in$ $\Omega$ and every positive $n$ $\in$ $\mathbb{Z}_+$, there exists a sequence of points $\mathbf{x}_i$ $\in$ $\Omega$, $0 \leq i \leq n$, such that $\mathbf{x}_0 = \mathbf{x}$, $\mathbf{x}_n =\bar{\mathbf{x}}$, and $|\mathbf{x}_i-\mathbf{x}_{i+1}| \leq \frac{C}{n} |\mathbf{x} - \bar{\mathbf{x}}|$ for all $i = 0,...,n-1$. Here $|\cdot |$ denotes the standard Euclidean norm.
\end{definition}

Note that every convex domain satisfies the chain condition. For a given real-valued function $a \in L^{\infty}(\Omega)$, $L^2_a(\Omega)$ refers to the set of all functions $f$ such that
$
\int_0^1|f(\mathbf{x})|^2a(\mathbf{x})d\mathbf{x}< \infty.
$
We will always assume that the associated function $a$ is uniformly bounded from below by a positive constant, in which case the space $L^2_a(\Omega)$ is a Hilbert space with respect to the weighted inner product $\langle \cdot , \cdot \rangle_{a} :  L^2_a(\Omega) \times L^2_a(\Omega) \rightarrow \mathbb{R}$, given by
$
\langle f, g\rangle_{a} = \int_{\Omega} f(\mathbf{x})\bar{g}(\mathbf{x})a(\mathbf{x})d\mathbf{x}
$
for each $f,g \in L^2_a(\Omega)$. We will also need the space $H^1_a(\Omega) = \big \lbrace z \in L^2_a(\Omega): (az)_{x_i} \in L^2(\Omega) \text{ for } 1 \leq i \leq N \big \rbrace$,
equipped with the norm
$
\|f\|_{H^1_a} = \Big( \|f\|^2_{a} + \sum_{i=1}^n\| (af)_{x_i}\|^2_{2}\Big)^{1/2} \nonumber
$.
When $a =\mathbf{1}$, where $\mathbf{1}$ is the function that takes the value $1$ a.e. on $\Omega$, the spaces $L^1(\Omega)$ and $H^1(\Omega)$ coincide with the spaces $L^1_a(\Omega)$ and $H^1_a(\Omega)$, respectively. We will also need the spaces 
$
W^{1,\infty}(\Omega) = \big \lbrace z \in L^{\infty}(\Omega): z_{x_i} \in L^{\infty}(\Omega) \text{ for } 1 \leq i \leq N \big \rbrace \nonumber
$
and 
$
W^{2,\infty}(\Omega) =  \big \lbrace z \in W^{1,\infty}(\Omega): z_{x_ix_i} \in L^{\infty}(\Omega) \text{ for } 1 \leq i \leq N \rbrace. \nonumber
$
Let $X$ be a Hilbert space with the norm $\|\cdot\|_X$. The space $C([0,T];X)$ consists of all continuous functions $u:[0,T] \rightarrow X$ for which $ \|u\|_{C([0,T];X)} ~:=~ \max_{0 \leq t \leq T} \|u(t)\|_X ~<~ \infty $. If $Y$ is a Hilbert space, then $\mathcal{L}(X,Y)$ will denote the space of linear bounded operators from $X$ to $Y$. We will also use the multiplication operator $\mathcal{M}_a : L^2(\Omega) \rightarrow L^2(\Omega)$, defined as $(\mathcal{M}_a u)(\mathbf{x}) = a(\mathbf{x}) u(\mathbf{x})$ for a.e. $\mathbf{x} \in \Omega$ and each $u \in L^2(\Omega)$.

We will need an appropriate notion of a solution of the PDE \eqref{eq:Mainsys1}. Toward this end, let $A$ be a closed linear operator that is densely defined on a subset $\mathcal{D}(A)$, the domain of the operator, of a Hilbert space $X$. We will define ${\rm spec}(A)$ as the set $\lbrace \lambda \in \mathbb{C}:\lambda\mathbb{I} -A $ is not invertible in $X \rbrace$, where $\mathbb{I}$ is the identity map on $X$. If $A$ is a bounded operator, then $\|A\|_{op}$ will denote the operator norm induced by the norm defined on $X$. From \cite{engel2000one}, we have the following definition.

\begin{definition} 
For a given time $T>0$, a \textbf{mild solution} of the ODE
\begin{equation} 
\label{eq:absode}
\dot{u}(t) = Au(t); \hspace{2mm} u(0) = u_0 \in X
\end{equation}
is a function $u \in C([0,T];X)$ such that $u(t) = u_0 + A\int_0^tu(s)ds$ for each $t \in [0, T]$. 
\end{definition}
Under appropriate conditions satisfied by $A$, the mild solution is given by a {\it strongly continuous semigroup} of linear operators, $(\mathcal{T}(t))_{t\geq 0}$, that are {\it generated} by the operator $A$ \cite{engel2000one}. 


The differential equations that we analyze in this paper will be non-autonomous in general. Hence, we must adapt the notion of a mild solution to these types of equations. 
\begin{definition}
Let $A_i$ be a closed linear operator with domain $\mathcal{D}(A_i)$ for each $i \in \mathbb{Z}_+$. Suppose that for a certain time interval $[0,T]$, a piecewise constant family of operators is given by a map, $t \mapsto A(t)$, for which there exists a partition $[0,T] = \cup_{i \in \mathbb{Z}_+} [a_i,a_{i+1})$ such that $a_i \leq a_{i+1}$ for each $i \in \mathbb{Z}_+$ and $A(t) = A_i$ for each $t \in [a_i,a_{i+1})$. Then a mild solution of the ODE
\begin{equation} 
\label{eq:nonaut}
\dot{u}(t) = A(t)u(t); \hspace{2mm} u(0) = u_0 \in X
\end{equation}
is a function $u \in C([0,T];X)$ such that 
\begin{equation} 
\label{eq:mldsol}
u(t) = u_0 + \sum_{i \in  \mathbb{Z}_+}A_i\int_{\min \lbrace t,a_{i} \rbrace}^{\min \lbrace t,a_{i+1} \rbrace} u(s)ds
\end{equation}
for each $t \in [0, T]$. 
\end{definition}
There is in fact a more general notion of mild solutions that arises from two-parameter semigroups of operators generated by time-varying linear operators. However, the definition \eqref{eq:mldsol} will be sufficient for our purposes, since one can construct solutions of the ODE \eqref{eq:nonaut} by treating it as an autonomous system in each time interval $[a_i,a_{i+1})$ and concatenating these solutions together to obtain the solution $u$. Note that the mild solution is defined with respect to an operator $A$ or  collection of operators $A(t)$; when we refer to such a solution, the associated operator(s) will be clear from the context. We will also need the notion of a positive semigroup, which is defined as follows.

\begin{definition}
A strongly continuous semigroup of linear operators $(\mathcal{T}(t))_{t \geq 0}$ on a Hilbert space $X$ is called \textbf{positive} if $u \in X$ such that $u \geq 0$ implies that $\mathcal{T}(t)u \geq 0$ for all $t \geq 0$.
\end{definition}

We introduce some additional notation from graph theory which will be used in the coming sections. We denote by $\mathcal{G} = (\mathcal{V},\mathcal{E})$ a directed graph with a set of $M$ vertices, $\mathcal{V} = \lbrace 1,2,...,M \rbrace$, and a set of $N_{\mathcal{E}}$ edges, $  \mathcal{E} \subset \mathcal{V} \times \mathcal{V}$. An edge from vertex $i \in \mathcal{V}$ to vertex $j \in \mathcal{V}$ is denoted by $e = (i,j) \in \mathcal{E}$. We define a source map $S : \mathcal{E} \rightarrow \mathcal{V}$  and a target map  $T: \mathcal{E} \rightarrow \mathcal{V}$ for which $S(e) = i$ and $T(e) = j$ whenever $e = (i,j) \in \mathcal{E}$. There is a {\it directed path} of length $s$ from vertex $i\in \mathcal{V}$ to vertex $j\in \mathcal{V}$ if there exists a sequence of edges $\{e_i\}^s_{i=1}$ in $\mathcal{E}$ such that $S(e_1) = i$, $T(e_s) = j$, and $S(e_k)=T(e_{k-1})$ for all $2 \leq k \leq s$. A directed graph $\mathcal{G}=(\mathcal{V},\mathcal{E})$ is called {\em strongly connected} if for every pair of distinct vertices $v_0,\,v_T\in \mathcal{V}$, there exists a {\em directed path} of edges in $\mathcal{E}$ connecting $v_0$ to $v_T$. We assume that $(i,i) \notin \mathcal{E}$ for all $i \in \mathcal{V}$. The graph  $\mathcal{G}$ is said to be {\it bidirected} if $e \in \mathcal{E}$ implies that $\tilde{e} = (T(e), S(e))$ also lies in $\mathcal{E}$.

\section{PROBLEM FORMULATION}
\label{prborm}
In this section, we will formulate the main problems that are addressed in this paper. The precise regularity conditions under which we study the solutions and parameters of the PDEs are presented in Section \ref{sec:secctrban}. On the other hand, the associated systems of stochastic differential equations (SDEs) will be treated in this paper only in a formal manner, in order to motivate the PDE controllability problem that we investigate. 
We note that one can prove existence of the stochastic processes that we consider by extending the semigroup-theoretic arguments in this paper to the theory of {\it Dirichlet forms} \cite{bass1991some}.

Consider a swarm of $N_p$ agents that are deployed on the $n$-dimensional domain $\Omega$. The position of each agent, indexed by $i \in \lbrace 1,2,...,N_p \rbrace$, evolves according to a stochastic process $\mathbf{Z}_i(t) \in \Omega$, {where $t$ denotes time. We assume that the agents are non-interacting. Therefore, the random variables that correspond to the dynamics of each agent are independent and identically distributed, and we can drop the subscript $i$ and define the problem in terms of a single stochastic process $\mathbf{Z}(t) \in \Omega$. The deterministic motion of each agent is defined by a velocity vector field $\mathbf{v}(\mathbf{x},t) \in \mathbb{R}^n$, where $\mathbf{x} \in \Omega$. This motion is perturbed by a $n$-dimensional {\it Wiener process} $\mathbf{W}(t)$, which models noise. This process can be a model for stochasticity arising from inherent sensor and actuator noise. Alternatively, noise could be actively programmed into the agents' motion 
to implement more exploratory agent behaviors and to take advantage of the smoothening effect of the process on the agents' probability densities. 
Given the parameter $\mathbf{v}(\mathbf{x},t)$, each agent evolves according to a {\it reflected diffusion process} $\mathbf{Z}(t)$ that satisfies the following SDE \cite{pilipenko2014introduction}: 
\begin{eqnarray}
\label{eq:SDE1}
d\mathbf{Z}(t) &=& ~ \mathbf{v}(\mathbf{Z}(t),t)dt + \sqrt{2D}d\mathbf{W}(t) + \mathbf{n}(\mathbf{Z}(t))d\psi(t), \nonumber \\ 
\mathbf{Z}(0) &=& ~ \mathbf{Z}_0,
\end{eqnarray}
where $\psi(t) \in \mathbb{R}$ is called the {\it reflecting function} or {\it local time} \cite{bass1991some,pilipenko2014introduction}, a stochastic process that constrains $\mathbf{Z}(t)$ to the domain $\Omega$; $\mathbf{n}(\mathbf{x})$ is the normal to the boundary at $\mathbf{x} \in \partial \Omega$; and $D>0$ is the diffusion constant. Without loss of generality, we assume that $D =1$.

We now pose the problem of determining the existence of the robot control law, defined as the velocity field $\mathbf{v}(\cdot,t)$, that drives the swarm to a target spatial distribution over the domain.
\begin{problem} 
\label{prob:SDEctrl}
Given a time $t_f>0$ and a target probability density $f:\Omega \rightarrow \mathbb{R}_{+}$ such that $\int_{\Omega} f(\mathbf{x})d\mathbf{x} = 1$, determine if there exists a feedback control law $\mathbf{v}: \Omega \times [0,t_f] \rightarrow \mathbb{R}^n$ such that the process \eqref{eq:SDE1} satisfies $\mathbb{P}(\mathbf{Z}(T) \in \Gamma) = \int_\Gamma f(\mathbf{x})d\mathbf{x}$ for each measurable subset $\Gamma \subset \Omega $.
\end{problem}

The {\it Kolmogorov forward equation} corresponding to the SDE  \eqref{eq:SDE1} is given by:
\begin{eqnarray}
\label{eq:Mainsys1} 
&y_t = \Delta y -  \nabla \cdot (\mathbf{v}(\mathbf{x},t) y) & in  ~~ \Omega \times [0,T] \nonumber \\ 
&y(\cdot,0) = y^0 & in ~~ \Omega \nonumber \\ 
& \mathbf{n}\cdot (\nabla y - \mathbf{v}(\mathbf{x},t) y)= 0  & in ~~ \partial  \Omega \times [0,T].
\end{eqnarray}
The solution $y(\mathbf{x},t)$ of this equation represents the probability density of a single agent occupying position $\mathbf{x} \in \Omega$ at time $t$, or alternatively, the density of a population of agents at this position and time. 
The PDE \eqref{eq:Mainsys1} is related to the SDE \eqref{eq:SDE1} through the relation $\mathbb{P}(\mathbf{Z}(t) \in \Gamma) = \int_{\Gamma} y(\mathbf{x},t)d\mathbf{x}$ for all $t \in [0,t_f]$ and all measurable $\Gamma \subset \Omega$. Therefore, the solution $y(\mathbf{x},t)$ captures the {\it mean-field} behavior of the population. In particular, as the number of agents tends to infinity, the {\it empirical measures} \cite{bensoussan2013mean}  converge to the measure for which this PDE's solution is the density $y(\mathbf{x},t)$. See \cite{zhang2017performance} for such a convergence analysis.
Problem \ref{prob:SDEctrl} can be reframed in terms of equation \eqref{eq:Mainsys1} as a PDE controllability problem as follows:
\begin{problem}
\label{ADEctrl1}
Given $t_f>0$, $y^0 :\Omega\rightarrow \bar{\mathbb{R}}_{+}$, and $f:\Omega \rightarrow \mathbb{R}_{+}$ such that $\int_{\Omega} y^0(\mathbf{x})d\mathbf{x} =\int_{\Omega} f(\mathbf{x})d\mathbf{x} = 1$, determine whether there exists a space- and time-dependent parameter $\mathbf{v}: \Omega \times [0,t_f] \rightarrow \mathbb{R}^n$ such that the solution $y$ of the PDE \eqref{eq:Mainsys1} satisfies $y(\cdot,t_f) = f$.
\end{problem}
%


In models of robotic swarms, it is useful to consider {\it hybrid} variants of the SDE \eqref{eq:SDE1} to account for the fact that each robot, in addition to a continuous spatial state $\mathbf{Z}(t)$, can be associated with a discrete state $Y(t) \in \mathcal{V} = \lbrace 1,...,N\rbrace$ at each time $t$ \cite{milutinovic2006modeling,elamvazhuthi2018optaut}. The elements of $\mathcal{V}$ can correspond to different behavioral states or tasks that can be performed by a robot, such as ``searching,'' ``lifting,'' or ``digging.'' In this case, the state of each agent is denoted by the pair $(\mathbf{Z}(t),Y(t)) \in \Omega \times \mathcal{V}$.
Suppose that the variable $Y(t)$ evolves according to a CTMC. We define a graph $\mathcal{G} = (\mathcal{V},\mathcal{E})$ in which the vertex set $\mathcal{V}$ is the set of discrete states, and the edge set $\mathcal{E}$ defines the possible agent transitions between the  discrete states in $\mathcal{V}$. The agents' transition rules are determined by the control parameters $u_{e}:[0,\infty) \rightarrow U$ for each $e \in \mathcal{E}$, also known as the {\it transition rates} of the associated CTMC. Here $U \subset \mathbb{R}_{+}$ is the set of {\it admissible transition rates}. 

The variable $Y(t)$ evolves on the state space $\mathcal{V}$ according to the conditional probabilities 
\begin{equation}
\label{eq:marko}
\mathbb{P}\left(Y(t+h) = T(e) ~|~ Y(t) = S(e)\right) =~ u_{e}(t)h + o(h)
\end{equation}
for each $e \in \mathcal{E}$.  Let  $\mathcal{P}(\mathcal{V}) = \lbrace \mathbf{y} \in \mathbb{R}^{N}_{+} : \sum_v y_v = 1 \rbrace$ be the simplex of probability densities on $\mathcal{V}$. Corresponding to the CTMC is a set of ODEs that determine the time evolution of the probability densities $\mathbb{P}(Y(t) = v) = \mu_v(t) \in \mathbb{R}_{+}$. The forward equation is given by a system of linear ODEs,
\begin{eqnarray}
\label{eq:ctrsys}
\dot{\boldsymbol{\mu}}(t) &=& \sum_{ e\in \mathcal E} u_e(t) \mathbf{Q}_e \boldsymbol{\mu}(t), \hspace{3mm} t \in [0, \infty), \\ \nonumber
\boldsymbol{\mu}(0) &=& \boldsymbol{\mu}^0 \in \mathcal{P}(\mathcal{V}),
\end{eqnarray}
where $\mathbf{Q}_e$ are control matrices whose entries are given by
\[
 Q_e^{ij} =
  \begin{cases}
   -1 & \text{if } i = j=  S(e),\\
    1 & \text{if } i= T(e), \hspace{1mm} j = S(e),\\
   0       & \text{otherwise.}
  \end{cases}
\]
Given these definitions, we can define a hybrid switching diffusion process $(\mathbf{Z}(t),Y(t))$ as a system of SDEs of the form
\begin{eqnarray}
d\mathbf{Z}(t) &\hspace{-2mm}=\hspace{-1mm}&  \mathbf{v}(Y(t),\mathbf{Z}(t),t)dt + \sqrt{2\mathbf{D}}d\mathbf{W}(t) + \mathbf{n}(\mathbf{Z}(t))d\psi(t), \nonumber \\ 
\mathbf{Z}(0) &\hspace{-2mm}=\hspace{-1mm}&  \mathbf{Z}_0, 
\label{eq:SDE2}
\end{eqnarray}
where $\mathbf{v}: \mathcal{V} \times \Omega \times [0,t_f] \rightarrow \mathbb{R}^n$ is the state- and time-dependent velocity vector field, and $\mathbf{D} \in \mathbb{R}^{N}_+$ is a vector of positive elements. Here, $D_k$ is the diffusion parameter associated with each discrete state $k \in \mathcal{V}$. Let $\mathbf{v}_k$ denote the velocity field associated with discrete state $k \in \mathcal{V}$.  Then the forward equation for this system of SDEs is given by the system of PDEs
\begin{eqnarray}
\nonumber
&(y_k)_t = D_k\Delta y_k -  \nabla \cdot (\mathbf{v}_k(\mathbf{x},t) y_k) +\mathcal{F}_k & in  ~~ \Omega \times [0,t_f] \\ \nonumber
&y_k(\cdot,0) = y^0_k & in ~~ \Omega  \\ \nonumber
& \mathbf{n}\cdot (\nabla y_k - \mathbf{v}_k(\mathbf{x},t) y_k)= 0  & in ~~ \partial  \Omega \times [0,t_f], \label{eq:Mainsys2} \\
\end{eqnarray}
where $k \in \mathcal{V}$ and $\mathcal{F}_k =\sum_{e \in\mathcal{E}} \sum_{j \in \mathcal{V}}u_e(t)Q^{kj}_e y_j$. 

We can pose a problem for the system of SDEs \eqref{eq:SDE2} that is similar to the one defined in Problem \ref{ADEctrl1}, with a target spatial distribution assigned to each discrete state:
\begin{problem}
\label{ADRctrl2}
Given $t_f>0$, $\mathbf{y}^0 :\Omega \rightarrow \mathbb{R}_{\geq 0}^{N}$, and $\mathbf{f}:\Omega^{N} \rightarrow \mathbb{R}_{+}^{N}$ such that $\sum_{i \in \mathcal{V}}\int_{\Omega} y^0_i(\mathbf{x})d\mathbf{x} = \sum_{i \in \mathcal{V}} \int_{\Omega} f_i(\mathbf{x})d\mathbf{x} = 1$, determine whether there exists a set of space- and time-dependent parameters $\mathbf{v}_k: \Omega \times [0,t_f] \rightarrow \mathbb{R}^n$ and time-dependent parameters $u_e:[0,t_f] \rightarrow \mathbb{R}_+$ such that the solution $\mathbf{y}$ of the system of PDEs  \eqref{eq:Mainsys2} satisfies $\mathbf{y}(\cdot,t_f) = \mathbf{f}$.
\end{problem}

It will be shown that the answer to Problem \ref{ADRctrl2} can be concluded from the answer to the controllability problem  defined in Problem \ref{ADEctrl1} and the {\it small-time controllability} properties of the system \eqref{eq:ctrsys}. Toward this end, we recall some controllability notions from nonlinear control theory \cite{bloch2015nonholonomic}.
\begin{definition}
Given $U \subset  \bar{\mathbb{R}}_{+}$ and $\boldsymbol{\mu}^0  \in \mathcal{P}(\mathcal{V})$, we define $R^{U}(\boldsymbol{\mu}^0,t)$ to be the set of all $\mathbf{y} \in  \mathcal{P}(\mathcal{V})$ for which there exists an admissible control, $\mathbf{u} = \lbrace u_e \rbrace_{e \in \mathcal{E}}$, taking values in $U$ such that there exists a trajectory of  system (\ref{eq:ctrsys}) with $\boldsymbol{\mu}(0) = \boldsymbol{\mu}^0$, $\boldsymbol{\mu}(t) = \mathbf{y}$. The \textbf{reachable set from $\boldsymbol{\mu}^0$ at time $t_f$} is defined as
\begin{equation}
R_{t_f}^U(\boldsymbol{\mu}^0) = \bigcup_{0 \leq t \leq t_f} R^U(\boldsymbol{\mu}^0,t).
\end{equation}
\end{definition}

\begin{definition}
\label{def1}
The system \eqref{eq:ctrsys} is said to be {\bf small-time locally controllable (STLC)} from an equilibrium configuration $\boldsymbol{\mu}^{eq} \in \mathcal{P}(\mathcal{V})$ if the set of reachable states $R_{t_f}^U(\boldsymbol{\mu}^{eq})$ contains a neighborhood of $\boldsymbol{\mu}^{eq} \in \mathcal{P}(\mathcal{V})$ in the subspace topology of $\mathcal{P}(\mathcal{V})$ (as a subset of $\bar{\mathbb{R}}^{N}_{+}$) for every $t_f>0$.
\end{definition}
Here, we have defined local controllability in terms of the subspace topology of $\mathcal{P}(\mathcal{V})$.  This is because the set $\mathcal{P}(\mathcal{V})$ is invariant for the system \eqref{eq:ctrsys} of controlled ODEs, and hence one cannot expect controllability to a full neighborhood of $\boldsymbol{\mu}^{eq}$. 

Using the above definitions, we will also consider the following problem.
\begin{problem} \label{stlcprb}
Given $\boldsymbol{\mu}^{eq} $, determine whether the system  \eqref{eq:ctrsys} is STLC from $\boldsymbol{\mu}^{eq}.$ 
\end{problem}

Lastly, we will consider the problem of stabilizing a target stationary distribution $\mathbf{f}^{eq}$ of the process \eqref{eq:SDE2} using time-independent control laws, which are more practical for implementation than time-dependent control laws: 
\begin{problem}
\label{stabprb}
Given $\mathbf{f}^{eq} : \Omega \rightarrow \mathbb{R}_+^{N}$, where $\mathbf{f}^{eq}=[f_1 ... f_{N}]^T$, determine whether there exist time-independent and possibly spatially-dependent parameters $\mathbf{v}_k:\Omega \rightarrow \mathbb{R}^n$, $K_e : \Omega \rightarrow \mathbb{R}_{+}$ such that the solution of the system
\begin{eqnarray}
\nonumber
&(y_k)_t = D_k\Delta y_k -  \nabla \cdot (\mathbf{v}_k(\mathbf{x}) y_k) +\mathcal{F}_k & in  ~~ \Omega \times [0,\infty) \\ \nonumber
&y_k(\cdot,0) = y^0_k & in ~~ \Omega  \\ \nonumber
& \mathbf{n}\cdot (\nabla y_k - \mathbf{v}_k(\mathbf{x}) y_k)= 0  & in ~~ \partial  \Omega \times [0,\infty), \label{eq:Mainsys3} 
\end{eqnarray}
where $k \in \mathcal{V}$ and $\mathcal{F}_k =\sum_{e \in\mathcal{E}} \sum_{j \in \mathcal{V}}K_e(\mathbf{x})Q^{kj}_e y_j$, satisfies $\lim_{t \rightarrow \infty}y_k(\cdot,t) \rightarrow f_k$ for each $k \in \mathcal{V}$.  
\end{problem}

A similar problem can be posed for the PDE system \eqref{eq:Mainsys1}. The solution to this problem for system \eqref{eq:Mainsys1}, at least for smooth target stationary distributions on bounded domains, can be inferred from existing literature on {\it Gibbs distributions} of Fokker-Planck equations evolving on compact manifolds without boundary \cite{stroock1993logarithmic}. In this paper, we relax the smoothness requirement on the achievable stationary distributions for system \eqref{eq:Mainsys1} (see Lemma \ref{cvglma1}) and also consider (Euclidean) domains with possibly non-smooth boundaries $\partial \Omega$. More importantly, we show in Theorem \ref{solstbpr} that we can construct solutions to Problem \ref{stabprb} by combining the result in Lemma \ref{cvglma1} with the solution to the problem  of stabilizing target stationary distributions of the finite-dimensional system \eqref{eq:ctrsys}, which we previously established in \cite{elamvaz2016lin}. 

\section{Controllability Analysis}
\label{sec:secctrban}

\subsection{Controllability of PDE System \eqref{eq:Mainsys1}}

In this section, we prove one of the main theorems of this paper. Specifically, we show that the PDE system \eqref{eq:Mainsys1} is controllable to a large class of sufficiently regular target probability densities. We first provide some new definitions that will be used in the subsequent analysis.

Given $a \in L^{\infty}(\Omega)$ such that $a \geq c$ for some positive constant $c$, and $\mathcal{D}(\omega_a) = H^1_a(\Omega)$, we define the {\it sesquilinear form} $\omega_{a}:\mathcal{D}(\omega_{a}) \times \mathcal{D}(\omega_{a}) \rightarrow \mathbb{C}$ as
\begin{equation}
\label{eq:Clpgenform1}
\omega_{a}(u,v) = \int_{\Omega}  \nabla ( a (\mathbf{x}) u(\mathbf{x})) \cdot \nabla ( a (\mathbf{x})\bar{v}(\mathbf{x}))d\mathbf{x}
\end{equation} 
for each $u \in \mathcal{D}(\omega_a)$. 
We associate with the form $\omega_a$ an operator $A_a :\mathcal{D}(A_a)  \rightarrow  L^2_a(\Omega)$, defined as $ A_au = v $ if $\omega_a(u,\phi) = \langle v , \phi \rangle_a $ for all $\phi \in \mathcal{D}(\omega_a)$ and for all $u \in \mathcal{D}(A_a) = \lbrace g \in \mathcal{D}(\omega_a): ~ \exists f \in L^2_a(\Omega) ~ \text{s.t.} ~  \omega_a(g,\phi)= \langle f, \phi \rangle_a ~ \forall \phi \in \mathcal{D}(\omega_a) \rbrace$. 

Similarly, given $a \in L^{\infty}(\Omega)$ such that $a \geq c$ for some positive constant $c$ and $\mathcal{D}(\sigma_a) = H^1_a(\Omega)$, we define the {\it sesquilinear form} $\sigma_{a}:\mathcal{D}(\sigma_{a}) \times \mathcal{D}(\sigma_{a}) \rightarrow \mathbb{C}$ as
\begin{equation}
\label{eq:Clpgenform2}
\sigma_{a}(u,v) = \int_{\Omega}  \frac{1}{a(\mathbf{x})} \nabla (a(\mathbf{x}) u(\mathbf{x})) \cdot \nabla ( a(\mathbf{x}) \bar{v}(\mathbf{x})) d\mathbf{x}
\end{equation} 
for each $u \in \mathcal{D}(\sigma_a)$. 
As for the form $\omega_a$, we associate an operator $B_a :\mathcal{D}(B_a)  \rightarrow  L^2_a(\Omega)$ with the form $\sigma_a$. We define this operator as 
$
B_au = v
$
if $\sigma_a(u,\phi) = \langle v , \phi \rangle_a $ for all $\phi \in \mathcal{D}(\sigma_a)$  
and for all $u \in \mathcal{D}(B_a) = \lbrace g \in \mathcal{D}(\sigma_a): \exists f \in L^2_a(\Omega) ~ \text{s.t.} ~ \sigma_a(g,\phi)= \langle f, \phi \rangle_a ~ \forall \phi \in \mathcal{D}(\sigma_a) \rbrace$. 

Note that, formally, $-A_{\mathbf{1}}=-B_{\mathbf{1}}$ is the Laplacian operator $\Delta (\cdot)$ with {\it Neumann boundary condition} $\left(\mathbf{n}\cdot (\nabla \cdot ~)= 0 \mbox{  ~in~  } \partial  \Omega \right)$. For general extension domains $\Omega$, the normal derivative might not make sense since it might not be true that $\mathcal{D}(A_{\mathbf{1}})$ is a subset of $H^2(\Omega)$ \cite{jerison1989functional}. Then, the Neumann boundary condition has to be interpreted in a {\it weak sense}.

Using the above definitions, we derive some preliminary results on the unbounded operators $-A_a$ and $-B_a$. The semigroups generated by these operators will each play an important role in the proof of controllability of system \eqref{eq:Mainsys1}. 

\begin{lemma}
\label{opprop1}
The operators $A_a:\mathcal{D}(A_a) \rightarrow L^2_a(\Omega)$  and $B_a:\mathcal{D}(A_a) \rightarrow L^2_a(\Omega)$ are closed, densely-defined, and self-adjoint. Moreover, the operators have a purely discrete spectrum. 
\end{lemma}
\begin{proof}
Consider the associated form $\omega_a$. This form is $\it{closed}$, i.e., the space $\mathcal{D} (\omega_a)$ equipped with the norm $\|\cdot\|_{\omega_a}$, given by 
$
\|u\|_{\omega_a} = (\|u\|^2_a+ \omega_a(u,u))^{1/2}
$
for each $u \in \mathcal{D}(\omega_a)$, is complete. This is true due to the fact that the multiplication map $u \mapsto  a \cdot u$ is an isomorphism from $H^1_a(\Omega)$ to $H^1(\Omega)$ and $H^1(\Omega)$ is a Banach space. Moreover, the space $H^1_a(\Omega)$ is dense in $L^2_a(\Omega)$. This follows from the inequality $\|au-av\|_2 \leq \|a\|_{\infty}\|u-v\|_2$ for each $u, v \in L^2(\Omega)$, the fact that the spaces $L^2_{\mathbf{1}}(\Omega)$ and $L^2_{a}(\Omega)$ are isomorphic, and the fact that the $H^1(\Omega)$ is dense in $L^2(\Omega)$. In addition, it follows from the definition of the form $\omega_a$ that $\omega_a$ is \textit{symmetric}, meaning that $\omega_a(u,v) = \overline{\omega_a (v,u)}$ for each $u, v \in \mathcal{D}(\omega_a)$. The form $\omega_a$ is also \textit{semibounded}, i.e., there exists $m \in \mathbb{R}$ such that $\omega_a(u,u) \geq m \|u\|^2_a$ for each $u \in \mathcal{D}(\omega_a)$. Particularly, this inequality is true for $m=0$ since $\omega_a(u,u)$ is non-negative for all $u \in \mathcal{D}(\omega_a)$. Hence, it follows from \cite{schmudgen2012unbounded}[Theorem 10.7] that the operator $A_a$ is self-adjoint. To establish the discreteness of the spectrum of $A_a$, we note that $H^1(\Omega)$ is compactly embedded in $L^2(\Omega)$ whenever $\Omega$ is an extension domain (Definition \ref{extn}). This implies that when $H^1_a(\Omega) = \mathcal{D}(\omega_a)$ is equipped with the norm $\|\cdot\|_{\omega_a}$, then it is  also compactly embedded in $L^2_a(\Omega)$. From \cite{schmudgen2012unbounded}[Proposition 10.6], it follows that $A_a$ has a purely discrete spectrum.

For the operator $B_a$, we only check that the form $\sigma_a$ is closed. The rest of the proof follows exactly the same arguments as the proof for $A_a$. To prove that the form $\sigma_a$ is closed, we need to prove that the space $\mathcal{D}(\sigma_a)$ equipped with the norm $\|\cdot\|_{\sigma_a}$, given by 
$
\|u\|_{\sigma_a} = (\|u\|^2_a+ \sigma_a(u,u))^{1/2}
$
for each $u \in \mathcal{D}(\sigma_a)$, is complete. Note that due to the lower bound $c$ on $a$, there exist constants $k_1, k_2 > 0$ such that 
\begin{eqnarray}
&\hspace{-2cm} k_1 \int_{\Omega}  \nabla ( a (\mathbf{x}) u(\mathbf{x})) \cdot \nabla ( a (\mathbf{x})\bar{u}(\mathbf{x}))d\mathbf{x} \nonumber  \\  
&~~\leq~  \int_{\Omega} \frac{1}{a (\mathbf{x})} \nabla ( a (\mathbf{x}) u(\mathbf{x})) \cdot \nabla ( a (\mathbf{x})\bar{u}(\mathbf{x}))d\mathbf{x} \nonumber \\  
&\leq~  k_2  \int_{\Omega}  \nabla ( a (\mathbf{x}) u(\mathbf{x})) \cdot \nabla ( a (\mathbf{x})\bar{u}(\mathbf{x}))d\mathbf{x}
\end{eqnarray}
 for all $u \in H^1_a(\Omega)$. From these inequalities, it follows that 
$
k_1 \|u\|_{H^1_a} \leq \|u\|_{\sigma_a} \leq k_2 \|u\|_{H^1_a}
$ for all $u \in \mathcal{D}(\sigma_a)=H^1_a(\Omega)$.
Hence, the form $\sigma_a$ is closed. Due to the symmetry and semiboundedness of this form, it follows from \cite{schmudgen2012unbounded}[Theorem 10.7] that the operator is self-adjoint. Since the norm $\|\cdot \|_{\sigma_a}$ is equivalent to the norm $\|\cdot\|_{H^1_a}$, the discreteness of the spectrum of $B_a$ again follows from \cite{schmudgen2012unbounded}[Proposition 10.6] due to the compact embedding of $H^1_a(\Omega)$ in $L^2_a(\Omega)$.
\end{proof}

\begin{corollary}
\label{CorII}
Consider the PDE
\begin{eqnarray}
&y_t = \Delta (a(\mathbf{x})y) &~~ in  ~~ \Omega \times [0,T] \nonumber \\ 
&y(\cdot,0) = y^0 &~~ in ~~ \Omega \nonumber \\ 
&\mathbf{n} \cdot \nabla( a (\mathbf{x})y) =0 &~~ in ~~ \partial \Omega \times [0,T].   
\label{eq:cllp1}
\end{eqnarray}

Let $y^0 \in L^2_a({\Omega})$. Then $-A_a$ generates a semigroup of operators $(\mathcal{T}^A_a(t))_{t \geq 0}$ such that the unique mild solution $y \in C([0,T];L^2_a(\Omega))$ of the above PDE exists and is given by $y(\cdot,t)=\mathcal{T}^A_a(t)y^0$ for all $t \geq 0$. Additionally, the semigroup $(\mathcal{T}^A_a(t))_{t \geq 0}$ is positive.
 Finally, if $\|\mathcal{M}_a  y^0  \|_{\infty} \leq 1$, then $\|\mathcal{M}_a \mathcal{T}^A_a(t)y^0  \|_{\infty} \leq 1$ for all $t\geq 0$. 
\end{corollary}
\begin{proof}
First, we note that the operator $-A_a$ is \textit{dissipative}, i.e., $\|(\lambda + A_a)u\|_a \geq \lambda \|u\|_a$ for all $\lambda>0$ and all $u \in \mathcal{D}(A_a)$, since $\omega_a(u,u) \geq 0$ for all  $u \in \mathcal{D}(\omega_a)$. Next, we note that $-A_a$ is self-adjoint, and hence the adjoint operator $-A^*_a$ is dissipative as well. It follows from a corollary of the {\it Lumer-Phillips theorem} \cite{engel2000one}[Corollary II.3.17] that $-A_a$ generates a semigroup of operators $(\mathcal{T}^A_a(t))_{t \geq 0}$ that solves the PDE \eqref{eq:cllp1} in the mild sense. 

Second, we establish the positivity of the semigroup. Toward this end, we note that the absolute value function $|\cdot|: \mathbb{R} \rightarrow \mathbb{R}$ is Lipschitz. Hence, it follows from \cite{ziemer2012weakly}[Theorem 2.1.11] that $v \in H^1(\Omega)$ implies that $|v| \in H^1(\Omega)$ whenever $v$ is only real-valued. This implies that if $u \in \mathcal{D}(\mathcal{\omega}_a)$, then $|{\rm Re}(u)| \in \mathcal{D}(\mathcal{\omega}_a)$, where ${\rm Re}(\cdot)$ denotes the real component of its argument. Then the  positivity of the semigroup follows from \cite{ouhabaz2009analysis}[Theorem 2.7]. 

To prove the last statement in the corollary, consider the closed convex set $C = \lbrace u \in L^2(\Omega); ~{\rm Re}( u) = u, ~ u(\mathbf{x}) \leq 1/a(\mathbf{x}) ~ {a.e.} 	~\mathbf{x} \in \Omega \rbrace$. The projection of a function $u \in L^2_a(\Omega)$ onto the set $C$ can be represented by the (nonlinear) operator $P$, given by $P u = {\rm Re}(u) \wedge 1/a = \frac{1}{2}{\rm Re} (u) + \frac{1}{2}|{\rm Re} (u)-1/a|$. If $u \in \mathcal{D}(\omega_a)$, then it follows from the chain rule \cite{ziemer2012weakly}[Theorem 2.1.11] that $\nabla (aPu) = \frac{1}{2}{\rm sign}({\rm Re}(au) - 1 ) \nabla ({\rm Re}(au))+\frac{1}{2}\nabla ({\rm Re}(au))$. Hence, it follows that $\omega_a(Pu,Pu) \leq   \omega_a(u,u)$ for all $u \in \mathcal{D}(\Omega_a)$. According to \cite{ouhabaz2009analysis}[Theorem 2.3], this implies that the set $C$ is invariant under the positive semigroup $(\mathcal{T}^A_a(t))_{t \geq 0}$;  therefore, we can conclude that if $\| \mathcal{M}_ay^0  \|_{\infty} \leq 1$, then $\|\mathcal{M}_a \mathcal{T}^A_a(t)y^0  \|_{\infty} \leq 1$ for all $t\geq 0$.
\end{proof}

Using the same arguments as in the proof of Corollary \ref{CorII}, we have the following result.

\begin{corollary}
\label{CorIII}
The operator $-B_a$ generates a semigroup of operators $(\mathcal{T}^B_a(t))_{t \geq 0}$ on $L^2_a(\Omega)$. If additionally $a \in W^{1,\infty}(\Omega)$ and $y^0 \in L^2_a(\Omega)$, then $y(\cdot,t)=\mathcal{T}^B_a(t)y^0$ is a mild solution of the PDE
\begin{eqnarray}
&y_t = \Delta y - \nabla \cdot(\frac {\nabla f(\mathbf{x})}{f(\mathbf{x})} ~ y) &~~ in  ~~ \Omega \times [0,T] \nonumber \\ 
&y(\cdot,0) = y^0 &~~ in ~~ \Omega \nonumber \\ 
&\mathbf{n} \cdot ( \nabla y - \frac {\nabla f(\mathbf{x})}{f(\mathbf{x})}y ) =0 &~~ in ~~ \partial \Omega \times [0,T],   
\label{eq:cllp2}
\end{eqnarray}
with $f = 1/a \in W^{1,\infty}(\Omega)$.  Moreover, the semigroup $(\mathcal{T}^B_a(t))_{t \geq 0}$ is positive.
\end{corollary}

When $f \in C^{\infty}(\bar{\Omega})$,} the representation of the operator $\Delta (\cdot) - \nabla \cdot(\frac {\nabla f(\mathbf{x})}{f(\mathbf{x})} ~ \cdot~)$ in the form $ \nabla \cdot( f \nabla (\frac{1}{f} ~ \cdot ~))$ is a well-known technique in the literature on Fokker-Planck equations for SDEs with drifts generated by potential functions \cite{stroock1993logarithmic}. In Corollary \ref{CorIII}, however, since $a$ is only once weakly differentiable and $\mathcal{D}(\sigma_a) = H^1_a(\Omega)$ (or equivalently, $H^1(\Omega)$), the operation $\Delta y$ is not admissible unless $a$ has additional regularity. Hence, the mild solution $y$ should be interpreted as the weak solution of the PDE \eqref{eq:cllp2} when $ a,f \in W^{1,\infty}(\Omega)$; i.e., it can be shown that $y$ satisfies
\begin{eqnarray}
\label{eq:wkPDE}
&& \hspace{-10mm} \langle y_t , \phi \rangle_{V^*,V} =  \sigma_{a}(u,\phi) \\
&& \hspace{-10mm} = \int_\Omega  \nabla y(\mathbf{x}, t) \cdot \nabla \phi (\mathbf{x}) d \mathbf{x} 
+  \int_\Omega   \frac{\nabla f(\mathbf{x})}{f(\mathbf{x})} y(\mathbf{x}, t) \cdot \nabla \phi (\mathbf{x}) d \mathbf{x} \nonumber
\end{eqnarray}
for all $\phi \in V$ and a.e. $t \in [0,T]$, where $V = H^1(\Omega)$ and $V^*$ is the dual space of $V$. Here, the second equality follows from the product rule (Theorem \ref{prodrule}) and the fact that $a,f \in W^{1,\infty}(\Omega)$. Note that in the weak formulation \eqref{eq:wkPDE}, the second-order differentiability of $f$ or $y$ is not required. That the mild solution of a linear PDE is also a weak solution follows from standard energy estimates and weak-topology arguments. See also \cite{ball1977strongly}.

Next, we establish that the semigroups $ (\mathcal{T}^A_a(t))_{t \geq 0}$ and $(\mathcal{T}^B_a(t))_{t \geq 0}$ are  {\it analytic} \cite{lunardi2012analytic}. Additionally, we will show some mass-conserving properties and long-term stability properties of these semigroups. 

\begin{lemma}
\label{cvglma1}
The semigroups $(\mathcal{T}^A_a(t))_{t \geq 0}$ and $(\mathcal{T}^B_a(t))_{t \geq 0}$ that are generated by the operators $-A_a$ and $-B_a$, respectively, are analytic. Additionally, these semigroups have the following {\it mass conservation property}: if $y^0 \geq 0$ and $\int_{\Omega}y^0(\mathbf{x})d \mathbf{x} =1$, then $\int_{\Omega} (\mathcal{T}^A_a(t)y^0)(\mathbf{x})d \mathbf{x} = \int_{\Omega} (\mathcal{T}^B_a(t)y^0)(\mathbf{x})d \mathbf{x} = 1 $ for all $t \geq 0$. Moreover, $0$ is a simple eigenvalue of the operators $-A_a$ and $-B_a$. Hence, if $y^0 \geq 0$ and $\int_{\Omega}y^0(\mathbf{x})d \mathbf{x} = \int_{\Omega}f(\mathbf{x})d \mathbf{x} = 1$, then the following estimates hold:
\begin{eqnarray}
\|\mathcal{T}^A_a(t)y^0-f\|_{a} &\leq& M_0 e^{-\lambda t}\|y^0-f\|_a, \label{eq:expcn11} \\ 
\|\mathcal{T}^B_a(t)y^0-f\|_{a} &\leq&  \tilde{M}_0 e^{-\tilde{\lambda} t}\|y^0-f\|_a \label{eq:expcn12}
\end{eqnarray}
for some positive constants $M_0, \tilde{M}_0, \lambda, \tilde{\lambda}$ and all $t \geq 0$.
\end{lemma}
\begin{proof}
The operators $A_a$ and $B_a$ are self-adjoint and positive semi-definite. Hence, their spectra lie in $[0,\infty)$. From this, it follows that the corresponding semigroups generated by $-A_a$ and $-B_a$ are analytic from the definition of an analytic semigroup \cite{lunardi2012analytic}[Chapter II]. Let $\int_\Omega y^0(\mathbf{x}) d \mathbf{x} = 1 $ such that $y^0 \in L^2_a(\Omega)$. Then $\int_{\Omega}(y(\mathbf{x},t)-y^0( \mathbf{x}))d \mathbf{x} =  -\int_\Omega A_a(\int_0^t y(\mathbf{x},s) ds)d \mathbf{x}= -\omega_a(\int_0^t y(\mathbf{x},s) ds,1/a)=0$ for all $t \geq 0$. Hence, the integral preserving property of the semigroups holds. 
For the exponential stability estimates \eqref{eq:expcn11} and \eqref{eq:expcn12}, we note that since the domain $\Omega$ is a connected bounded extension domain, it follows from Poincar{\'e}'s inequality \cite{leoni2009first}[Theorem 12.23] that there exists a constant $C>0$ such that for all $u \in H^1(\Omega)$, $ \int_{\Omega}|u(\mathbf{x})-u_\Omega|^2 d\mathbf{x} ~\leq~ C \int_{\Omega} |\nabla u (\mathbf{x})|^2d\mathbf{x}$, 
where $u_\Omega = \frac{1}{\mu(\Omega)}\int_{\Omega} u(\mathbf{x})d\mathbf{x}$.
This implies that $0$ is a simple eigenvalue of the Neumann Laplacian operator $A_\mathbf{1}$. 
Since the operator $A_a$ can be written as a composition of operators $A_\mathbf{1}\mathcal{M}_a$, where $\mathcal{M}_a$ is the multiplication map $u \mapsto au$ from $H^1_a(\Omega)$ to $H^1(\Omega)$, it follows that $0$ is also a simple eigenvalue of $A_a$ with the corresponding eigenvector $f = 1/a$. Additionally, for a given $u \in H^1_a(\Omega)$, $\omega_a (u,u) = 0$ iff $\sigma_a (u,u) = 0$ due to the assumed positive lower bound on $a$. Hence, it also holds that $0$ is a simple eigenvalue of the operator $B_a$. Then we can derive the estimates \eqref{eq:expcn11} and \eqref{eq:expcn12} by applying \cite{engel2000one}[Corollary V.3.3] as follows. We note that $\mathcal{D}(A_a) \subset H^1_a(\Omega)$ and $\mathcal{D}(B_a) \subset H^1_a(\Omega)$ are compactly embedded in $L^2_a(\Omega)$, since $H^1_a(\Omega)$ is compactly embedded in $L^2_a(\Omega)$. Thus, since the semigroups $\mathcal{T}^A_a(t)$ and $\mathcal{T}^B_a(t)$ are analytic, which implies that $\mathcal{T}^A_a(t) \in \mathcal{D}(A_a)$ and $\mathcal{T}^B_a(t) \in \mathcal{D}(B_a)$ for all $y^0 \in L^2_a(\Omega)$ and all $t>0$, these semigroups are immediately compact. Moreover, the residue of the first eigenvalue $0$ in \cite{engel2000one}[Corollary V.3.3] is the projection operator onto the subspace spanned by the element $f \in L^2_a(\Omega)$, since $0$ is a simple eigenvalue with the corresponding eigenvector $f$.
\end{proof}

The above result implies that if $\mathbf{v}(\cdot, t) = \nabla f/f$, then the solution of system \eqref{eq:Mainsys1} exponentially converges to $f$ in the $L^2_a$-norm if $f$ is in $W^{1,\infty}(\Omega)$ and is bounded from below by a positive constant. Hence, this choice of $\mathbf{v}(\cdot, t)$ is a possible control law for achieving exponential stabilization of desired probability densities. In the next few results, we verify some regularizing properties of the semigroups considered above, which will be critical to our controllability analysis.

\begin{lemma}
\label{slblma1}
Let $a \in L^{\infty}(\Omega)$ be real-valued and uniformly bounded from below by a positive constant $c_1$. Moreover, let $y^0 \in L^2_a(\Omega)$ such that  $y^0 \geq c_2$ for some positive constant $c_2$. If $(\mathcal{T}^F_a(t))_{t \geq 0}$ is the semigroup generated by the operator $-A_a ~ {\rm or} ~ -B_a$, then $\mathcal{T}^F_a(t)y^0 \geq \frac{c_1 c_2 }{\|a\|_{\infty}}$ for all $t \geq 0$.
\end{lemma}
\begin{proof}
Let $k = c_1 c_2$. Then we know that $a \cdot y^0 \geq k$. Hence, we can decompose the initial condition as $y^0 = kf +(y^0 -kf)$, where $f = 1/a$. Note that $y^0 - kf$ is positive and $\mathcal{T}^B_a(t) f = f$ for all $t \geq 0$. Therefore, it follows from the positivity preserving property of the semigroup (Corollary \ref{CorIII}) that $\mathcal{T}^B_a(t) y^0 \geq k / \|a\|_{\infty}$ for all $t \geq 0$.
\end{proof}

	 In order to prove the next result (Proposition \ref{samedom}), we will use the following well-known product rule for Sobolev functions.
\begin{theorem}
	\label{prodrule}
	\textbf(Product Rule) Let $\Omega \subset \mathbb{R}^n$ be an open bounded set. Suppose that $u \in H^1(\Omega)$ and $v \in W^{1,\infty}(\Omega)$. Then $u \cdot v \in H^1(\Omega)$, and the weak derivatives of the product $u \cdot v$ are given by $(uv)_{x_i} = u_{x_i}v +v_{x_i}u$ for each $i \in \lbrace 1,...,n\rbrace$.
\end{theorem}

\begin{proposition}
\label{samedom}
Let $a \in W^{1,\infty}(\Omega)$ and uniformly bounded from below by a positive constant $c$. Then $\mathcal{D}(A_a) =\mathcal{D}(B_a)$. 
\end{proposition}
\begin{proof}
Let $u \in \mathcal{D}(B_a)$. Then using the product rule (Theorem \ref{prodrule}), we have that
$
\omega_a(u , \frac{\phi}{a}) = \big < B_au,\phi \big>_a - \big < \frac{1}{a^2}\nabla a \cdot \nabla (au), \phi  \big >_a
$
for all $\phi \in H^1_a(\Omega)$. Since $a$ is in $W^{1,\infty}(\Omega)$ and is bounded from below by a positive constant, $H^1(\Omega) = H^1_a(\Omega)$. Hence, $\phi  \in H^1_a(\Omega)$ implies that $a \cdot u, ~\frac{\phi}{a} \in H^1(\Omega)$ due to the product rule (Theorem \ref{prodrule}). Therefore, we can conclude that 
$
\omega_a(u , \phi) = \big < a \cdot B_au,\phi \big>_a - \big < \frac{1}{a}\nabla a \cdot \nabla (au), \phi  \big >_a
$ 
for all $\phi \in H^1_a(\Omega)$. Hence, $u \in \mathcal{D}(B_a)$ implies $u \in \mathcal{D}(A_a)$. To establish that $u \in \mathcal{D}(A_a)$ implies $u \in \mathcal{D}(B_a)$, we can use a similar argument: if $u \in \mathcal{D}(A_a)$, then
$
\sigma_a(u,a \phi) = \big < A_au,\phi \big >_a + \big < \frac{1}{a}\nabla a \cdot \nabla (a u), \phi \big >_a
$
for all $\phi \in H^1_a(\Omega)$.  
\end{proof} 

    
    This lemma will play an important role in the theorem on controllability, Theorem \ref{maintheo}. It will be used to conclude that solutions of the parabolic systems \eqref{eq:cllp1} and \eqref{eq:cllp2} have bounded gradients for each $t>0$, provided that the boundary of the domain $\Omega$ is regular enough. This will enable us to prove later on that the control inputs constructed to prove controllability are bounded.

\begin{lemma}
\label{egdom}
Let $a \in W^{1,\infty}(\Omega)$. Let $\Omega$ be a domain that is either $C^{1,1}$ or convex. Then
there exists $m \in \mathbb{Z}_+$ large enough such that, for some $C_m>0$, $\|\nabla(a(\mathbf{x})u)\|_{\infty} \leq C_m (\|(\mathbb{I} +A_a)^{m}u\|_{a})$ for all $u \in \mathcal{D}(A^m_a)$.
Similarly, there exists $m' \in \mathbb{Z}_+$ large enough such that, for some $C_{m'}>0$,  $\|\nabla(a(\mathbf{x})u)\|_{\infty} \leq C_{m'} (\|(\mathbb{I} +B_a)^{m'}u\|_{a})$ for all $u \in \mathcal{D}(B^{m'}_a)$.
\end{lemma}
\begin{proof}
First, we consider the case where $\Omega$ is a $C^{1,1}$ domain. Let $W^{2,p}(\Omega)$ be the set of elements in $L^p(\Omega)$ with second-order weak derivatives in $L^{p}(\Omega)$. Then we know that for the Neumann problem with $a = \mathbf{1}$, 
\begin{equation}
\label{eq:Lpreg}
-\Delta u + a_0u = f ~~ \text{in} ~ \Omega, ~~~
\mathbf{n} \cdot \nabla u =0  ~~ \text{in} ~ \partial \Omega
\end{equation}
has solutions $u \in W^{2,p}(\Omega)$ if $f \in L^p(\Omega)$ whenever $1< p < \infty$, $a_0 \in L^{\infty}(\Omega)$ \cite{grisvard2011elliptic}[Theorem 2.4.2.7]) and $a_0 \geq \beta$ for some $\beta >0$. These solutions have bounds
\begin{equation}
\label{eq:Lpregest}
\|u\|_{W^{2,p}} \leq C_{p}\|f\|_{L^p}
 \end{equation}
for some constant $C_p>0$. Note that $u \in L^p(\Omega)$ implies that $\mathcal{M}_au \in L^p(\Omega) $ for each $2 \leq p \leq \infty$. Suppose that $n > 2$, where we recall that $n$ is the dimension of the Euclidean space $\mathbb{R}^n$ of which $\Omega$ is a subset. Note that $\Omega$ is an extension domain  (Definition \ref{extn}). Then from the {\it $W^{2,p}$ regularity estimate} \eqref{eq:Lpregest} of equation \eqref{eq:Lpreg} and from the embedding theorem \cite{leoni2009first}[Corollary 11.9], it follows that $f \in L^2(\Omega)$ implies $(A_a+\mathbb{I})^{-m}f = (A_{\mathbf{1}}\mathcal{M}_a +\mathbb{I})^{-m}f = \big ( (A_{\mathbf{1}}+\mathcal{M}_a^{-1})\mathcal{M}_a\big )^{-m}f  \in L^q(\Omega)$ for some $q \geq n$ for $m \in \mathbb{Z}_+$ large enough. For the general case $n \geq 1$, it follows from the embedding theorem \cite{leoni2009first}[Theorem 11.23] that $f \in L^2(\Omega)$ implies $(A_a+\mathbb{I})^{-m}f  \in L^q(\Omega)$ for any desired $  n \leq q < \infty$, provided $m \in \mathbb{Z}_+$ for $m$ large enough. Since $a \in W^{1,\infty}(\Omega)$, it follows from the $W^{2,p}$ regularity estimate \eqref{eq:Lpregest} of equation \eqref{eq:Lpreg} and {\it Morrey's inequality} \cite{leoni2009first}[Theorem 11.34] that if $f \in L^2(\Omega)$, $(A_a+\mathbb{I})^{-m}f = u \in L^q(\Omega)$, and $m \in \mathbb{Z}_{+}$ is large enough, then $\|\nabla u\|_{\infty} \leq C_{\infty} \|f\|_2$, where $C_{\infty}>0$ is independent of $f$. 

A similar argument can be used when $\Omega$ is convex. However, it is not clear if the $W^{2,p}$ regularity estimate \eqref{eq:Lpregest} holds for general convex domains. On the other hand, it can be established that the {\it $L^p$ regularity estimate} of the PDE \eqref{eq:Lpreg} holds for such domains. In particular, it follows from \cite{bakry2013analysis}[Corollary 6.3.3] and the embedding theorems \cite{leoni2009first}[Corollary 11.9, Theorem 11.23] that for any $1 < p \leq \infty$, there exists $m \in \mathbb{Z}_{+}$ large enough such that $(A_a+\mathbb{I})^{-m}$ is a bounded operator from $L^2(\Omega)$ to $L^p(\Omega)$. This last statement uses only the extension property of the domain $\Omega$, which is not required to be convex for the statement to hold true; the convexity of $\Omega$ is required mainly to derive the bounds on the gradient of the solution $u$. For this derivation, we use a result from \cite{maz2009boundedness}. Since $a \in W^{1,\infty}(\Omega)$, it follows from the theorem  \cite{maz2009boundedness}[Theorem] that there exists a constant $C'_{\infty}>0$ such that if $f \in L^2(\Omega)$, $(A_a+\mathbb{I})^{-m}f = u$, and $m \in \mathbb{Z}_{+}$ is large enough,  then $\|\nabla u\|_{\infty} \leq C'_{\infty} \|f\|_2$, where $C'_{\infty}$ is independent of $f$. In this last statement, the theorem \cite{maz2009boundedness}[Theorem] can be applied to derive the gradient bounds due to the fact that $(A_a+\mathbb{I})^{-{m+1}}f \in L^q(\Omega)$ for some $q>n$ for $m \in \mathbb{Z}_{+}$ large enough.

The inequality for the operator $B_a$
can be derived using exactly the same approach as for $A_a$. Hence, we only point out the key results needed. Particularly, for a $C^{1,1}$ domain, the $W^{2,p}$ regularity estimate \eqref{eq:Lpregest} also holds for the equation $-\nabla \cdot (\frac{1}{a(\mathbf{x})} \nabla u) + a_0u =f$ from \cite{grisvard2011elliptic}[Theorem 2.4.2.7]. For the case where $\Omega$ is convex, the $W^{1,p}$ regularity estimate has been proved in \cite{geng2018homogenization}[Theorem 1.3] for solutions of elliptic operators in divergence form on convex domains. Since $a \in W^{1,\infty}(\Omega)$, using the product rule (Theorem \ref{prodrule}), the gradient bounds of the Neumann Laplacian \cite{maz2009boundedness}[Theorem] also give the desired gradient bounds of the operator $\nabla \cdot (\frac{1}{a(\mathbf{x})} \nabla \cdot)$.
\end{proof}

\begin{lemma}
\label{cvglma2}
Let $\Omega$ be a domain that is either $C^{1,1}$ or convex.
Let $y^0 \in \mathcal{D}(A^m_a)$ for some $m \in \mathbb{Z}_+$. Then the mild solution, $y \in C([0,\infty);L^2_a(\Omega))$, of the PDE \eqref{eq:cllp1} satisfies $y(\cdot,t) \in \mathcal{D}(A^m_a)$ for each  $t \in [0,\infty)$. Moreover, the following estimates hold for some  positive constants $M_m$ and $\lambda$:
\begin{eqnarray}
\|(\mathbb{I}+A_a)^m(y^0-f)\|_{a} ~\leq~ M_m e^{-\lambda t}. 
 \label{eq:expcn2}
\end{eqnarray}
\end{lemma}
\begin{proof}
We are given that $y^0 \in \mathcal{D}(A_a^m)$. Since the semigroup $(\mathcal{T}^A_a(t))_{t\geq 0}$ and its generator $-A_a$ commute, we know that $ \|(\mathbb{I}+A_a)^m(\mathcal{T}^A_a(t))(y^0-f)\|_a = \| \mathcal{T}^A_a(t)(\mathbb{I}+A_a)^m (y^0-f)\|_a \leq M_0 e^{-\lambda t}\|(\mathbb{I}+A_a)^m(y^0-f)\|_{a}$ for some positive constants $M_0$ and $\lambda$. 
\end{proof}


Since controllability will first be proved in Lemma \ref{prlma} under the assumption that the initial condition is bounded from below by a positive constant, the following lemma will be used to relax this assumption in Theorem \ref{maintheo}. 

\begin{lemma}
\label{Dave}
Let $\Omega$ be a domain that is either $C^{1,1}$  or convex and that satisfies the chain condition (see Definition \ref{chain}). Let $y^0 \in L^2(\Omega)$ be such that $y^0 \geq 0$. Let $y \in C([0,T];L^2(\Omega))$ be the unique mild solution of the PDE \eqref{eq:cllp1}. Then for all $t \in (0,\infty)$, there exists a positive constant, $c_t>0$, such that $y(\cdot,t) \geq c_t$.
\end{lemma}
\begin{proof}
Consider the heat equation with Neumann boundary condition, that is, the PDE \eqref{eq:Mainsys1} with $\mathbf{v} \equiv \mathbf{0}$.
The solution $y$ of this PDE can be represented using the {\it Neumann heat kernel} $K$. That is, there exists a measurable map $K:(0,\infty) \times \Omega^2   \rightarrow [0,\infty)$ such that the mild solution $y$ can be constructed using the relation $
y(\mathbf{x},t) =  \int_\Omega K(t,\mathbf{x},\mathbf{z})y^0(\mathbf{z})d\mathbf{z}
$
for each $t \in (0,\infty)$ and almost every $\mathbf{x} \in \Omega$. From \cite{choulli2015gaussian}[Theorem 3.1] (for $C^{1,1}$ domains) and
 \cite{li1986parabolic}[Corollary 2.1] (for convex domains),  for some $C>0$, we know that $
K(t,\mathbf{x},\mathbf{z}) \geq \frac{C}{(4\pi t)^{1/2}}\exp(\frac{-|\mathbf{x}-\mathbf{z}|^2}{4t})
$
for each $t>0$ and almost every $\mathbf{x},\mathbf{z} \in \Omega$. From this, the lower bound on $y(\cdot,t)$ follows. 
\end{proof}

In the following lemma, we establish a controllability result under the assumption that the initial condition is sufficiently regular and the final time $t_f$ is equal to the infinite summation $\sum_{k=1}^{\infty}\frac{1}{k^2}$. 

\begin{lemma}
\label{prlma}
Let $y^0 \in \mathcal{D}(A^m_a)$ for some $m \in \mathbb{Z}_+$  such that $y^0 \geq c_1 $ for some positive constant $c_1$. Suppose that $f \in W^{1, \infty}(\Omega)$ such that $f\geq c_2$  for some positive constant $c_2$ and $\int_{\Omega}f(\mathbf{x})d\mathbf{x} = \int_{\Omega}y^0(\mathbf{x})d\mathbf{x}$. Assume that the final time is specified as  $t_f = \sum_{k=1}^{\infty} \frac{1}{k^2}$. Define the vector field $\mathbf{v}$ in the PDE \eqref{eq:Mainsys1} by
\begin{equation}
\label{eq:fbctrl}
\mathbf{v}(\cdot,t) =   \frac{\nabla y}{y}- \alpha  j \frac{\nabla (a y)}{y}
\end{equation}
for some $\alpha > 0$, $j \in \mathbb{Z}_+$, where $a = 1/f$ whenever $t \in [\sum_{k=1}^{j-1} \frac{1}{k^2},\sum_{k=1}^{j} \frac{1}{k^2})$. Here, we define $\sum_{k=1}^{j} \frac{1}{k^2} = 0$ if $j=0$.

If $\Omega$ is a domain that is $C^{1,1}$ or convex, then there exist $m' \in \mathbb{Z}_+ $ and $\alpha >0$ large enough and independent of $m$ such that if $m \geq m'$, then $\mathbf{v} \in  L^\infty([0,t_f];L^\infty(\Omega)^n)$ and $y(\cdot,t_f)=f$.
\end{lemma}
\begin{proof}
Substituting $\mathbf{v}(\cdot,t) = \frac{\nabla y}{y}- \alpha j \frac{\nabla (a y)}{y}$ whenever $t \in [\sum_{k=1}^{j-1} \frac{1}{k^2},\sum_{k=1}^{j} \frac{1}{k^2})$ in the PDE \eqref{eq:Mainsys1}, it can be seen that if the solution of this PDE exists, then it can be constructed from mild solutions of the {\it closed-loop} PDE
\begin{eqnarray}
&\tilde{y}_t = \alpha j \Delta (a(\mathbf{x})\tilde{y}) &~~ in  ~~ \Omega \times [0,\frac{1}{j^2}) \nonumber \\ 
&\tilde{y}(\cdot,0) = \tilde{y}^0=y(\cdot,\sum_{k=1}^{j-1} \frac{1}{k^2}) &~~ in ~~ \Omega \nonumber \\ 
&\mathbf{n} \cdot \nabla (a\tilde{y}) = 0    &~~ in ~~ [0,\frac{1}{j^2}). \label{eq:pssol}
\end{eqnarray}
In particular, if $t \in [\sum_{k=1}^{j-1} \frac{1}{k^2},\sum_{k=1}^{j} \frac{1}{k^2})$, then $t = i+\sum_{k=1}^{j-1} \frac{1}{k^2}$ for some $i \in [0,\frac{1}{j^2})$. Hence, we obtain the relation $y(\cdot, t)=y(\cdot, \sum_{k=1}^{j-1} \frac{1}{k^2} + i) = \tilde{y}(\cdot,i)$. This implies that $y(\cdot, \sum_{k=1}^{j-1} \frac{1}{k^2} + i) = \tilde{y}(\cdot,i)$ for each $i \in [0,\frac{1}{j^2})$ and each $m \in \mathbb{Z}_+$. Since $y^0$ is uniformly bounded from below by a positive constant, $\tilde{y}$ is also uniformly bounded from below according to Lemma \ref{Dave}. Moreover, since $a \in W^{1,\infty}(\Omega)$, it follows that $\mathcal{D}(A_a) \subset H^1_a(\Omega)$. Due to the regularity assumption on the initial condition, we know that $y(\cdot,t) \in \mathcal{D}(A_a)$ for all 
$t \in [0,t_f)$. Hence, the velocity field $\mathbf{v}$ is well-defined for all $t \in [0,t_f)$, i.e., it is at least square-integrable at almost every time $t$.
 
Noting that $-\beta A_a$ generates the rescaled semigroup $(\mathcal{T}^A_a(\beta t))_{t \geq 0}$, and applying Lemma \ref{cvglma1}, it follows that
$ \|y(\cdot,\sum_{k=1}^{j}\frac{1}{k^2})-f\|_{a} = \|\mathcal{T}^A_a(\alpha k/k^2)...\mathcal{T}^A_a(\alpha)y^0 -f\|_{a} = \|\mathcal{T}^A_a(\alpha \sum_{k=1}^{j}\frac{1}{k})y_0 -f\|_{a}  \leq  
 M_0 e^{-\alpha \lambda \sum_{k=1}^{j}\frac{1}{k}}
$ for each $j \in \mathbb{Z}_+$, for some positive constants $M_0$ and $ \lambda$ independent of $j$. Since the summation $\sum_{k=1}^{j}\frac{1}{k}$ is diverging, we have that $y(\cdot,t_f) = f$ if the solution is defined over the interval $[0,t_f]$. Since $y$ is continuous on $[0,t_f)$ and uniformly bounded, it follows that $y$ is in $C([0,t_f);L^2_a(\Omega))$ and can be extended to a unique mild solution $y \in C([0,t_f];L^2_a(\Omega))$ defined over the time interval $[0,t_f]$. 
 
It is additionally required to prove the existence of $m' \in \mathbb{Z}_+$ and $\alpha >0$ such that, if $m\geq m'$, then $\mathbf{v} \in  L^\infty([0,t_f];L^{\infty}(\Omega)^n)$. First, we derive bounds on the term 
$1/y(\cdot,t)$. Due to the lower bound on the initial condition $y^0$, and noting that $y(\cdot,t) = \mathcal{T}^A_a(\tilde{t})y^0$ for some $\tilde{t} \in [0,\infty)$ depending on $t \in [0,t_f)$, it follows from Lemma \ref{slblma1} that there exists a positive constant $d$ such that
\begin{equation}
\label{eq:bd1}
y(\cdot,t) \geq d
\end{equation}
 for all $t\in [0,t_f)$. This gives us the uniform upper bound $1/d$  on the term $1/y(\cdot,t)$.
Next, we consider the term  $\alpha \nabla (a(\mathbf{x})y(\cdot, t))$. We note that $y_{0} \in \mathcal{D}(A^j_a)$. Hence, we can apply the estimate in Lemma \ref{cvglma2} to obtain $
\|(\mathbb{I}+A_a)^m(y(\cdot,\sum_{k=1}^{m}\frac{1}{k^2})-f)\|_{a} ~\leq~ \tilde{M}e^{-\alpha \lambda \sum_{k=1}^{j}\frac{1}{k}} 
$ for some positive constants $\tilde{M}$ and $\lambda$. From Lemma \ref{egdom}, it follows that when $\Omega$ is a domain that is $C^{1,1}$  or convex, there exist $C>0$ and $m'\in \mathbb{Z}_+$ depending only on $a$ such that, if $m \geq m'$, then
\begin{equation}
\label{eq:bd2}
 \|\alpha j \nabla (a(\mathbf{x})y)(\cdot,\sum_{k=1}^{j}\frac{1}{k^2})\|_{\infty} ~\leq~ C\alpha j\tilde{M}e^{- \alpha \lambda \sum_{k=1}^{j}\frac{1}{k}}
\end{equation}
for some positive constants  $\tilde{M}$ and $\lambda$.
The right-hand side of the estimate \eqref{eq:bd2} is not uniformly bounded for arbitrary $\alpha>0$ due to its dependence on $j$. However, we note that $\lim_{j \rightarrow \infty} -\mbox{ln} ~j  + \sum_{k=1}^{j}\frac{1}{k} = \gamma$, where $\gamma>0$ is the {\it Euler-Mascheroni} constant \cite{finch2003mathematical}[Section 1.5]. Therefore, by setting $\alpha  \geq 1/ \lambda$, the right-hand side becomes uniformly bounded for all $j \in \mathbb{Z}_+$. Since $a \in W^{1,\infty}(\Omega)$, it follows from the product rule  and the estimate \eqref{eq:bd2} that 
\begin{equation}
\label{eq:bd3}
 \| \nabla y(\cdot,t)\|_{\infty} ~\leq~ C_2
\end{equation}
for some positive constant $C_2$ and for all $t \in [0,t_f)$.

From the estimates \eqref{eq:bd1}-\eqref{eq:bd3}, it follows that if $\alpha>0$ is large enough, then $\mathbf{v} \in L^\infty([0,t_f];L^\infty(\Omega)^n)$ and $y(\cdot,t_f)=f$. This concludes the proof for the case when the domain $\Omega$ is  $C^{1,1}$  or convex. 
 \end{proof}

\vspace{1mm}

Note that any control law of the form $\mathbf{v}(\cdot,t) =   \frac{\nabla y}{y}- \alpha m^\beta \frac{\nabla (a y)}{y}$ for numerous other values of $\beta$ and $\alpha$ will also achieve the desired controllability objective, due to the fact that an exponential function of a variable grows faster than a polynomial function as the variable tends to infinity. Additionally, we could replace the parameter $m$ with a continuous function $m(t)$ such that $\int_0^T m(\tau)d\tau = \infty$.

The following corollary follows from  Lemma \ref{prlma} using a straightforward scaling argument.

\begin{corollary}
\label{prlma2}
Let $y^0 \in \mathcal{D}(A^m_a)$ be such that $y^0 \geq c_1 $ for some positive constant $c_1$ and $m \in \mathbb{Z}_{+}$. Let $\Omega$ be a domain that is either $C^{1,1}$ or convex. Suppose that $f \in W^{1, \infty}(\Omega)$ such that $f\geq c_2$ for some positive constant $c_2$, $\int_\Omega f(\mathbf{x})d\mathbf{x} = \int_\Omega y^0(\mathbf{x})d\mathbf{x}$, and $a = 1/f$. Let $t_f >0$ be the final time. Then if $m$ is large enough there exists $\mathbf{v} \in L^{\infty}([0,t_f];L^{\infty}(\Omega)^n)$ such that  the mild solution $y$ of the PDE  \eqref{eq:Mainsys1} satisfies $y(\cdot,t_f) = f$.  
\end{corollary}

Now, we are ready to state and prove our main theorem, where we relax the assumptions on the initial condition $y^0$ made in Corollary \ref{prlma2}. However, we will need to impose the additional constraint that $\Omega$ should satisfy the chain condition.

\begin{theorem}
\label{maintheo}
Let $\Omega$ be a domain that is either  $C^{1,1}$ or convex and that satisfies the chain condition. Let $y^0 \in L^2(\Omega)$ be such that $y^0 \geq 0 $ and $\int_{\Omega} y^0(\mathbf{x})d\mathbf{x}=1$. Suppose that $f \in W^{1,\infty}(\Omega)$ such that $f\geq c$ for some positive constant $c$, $\int_{\Omega} f(\mathbf{x})d\mathbf{x}=1$. Let $t_f>0$ be the final time. Then there exists $\mathbf{v} \in L^\infty([0,t_f];L^\infty(\Omega)^n)$ such that the unique mild solution $y$ of the PDE \eqref{eq:Mainsys1} satisfies $y(\cdot,t_f)=f$. 
\end{theorem}
\begin{proof}
Set $\mathbf{v}(\cdot,t) = \mathbf{0}$ in the PDE \eqref{eq:Mainsys1} for each $t \in [0, \epsilon/3]$, where $\epsilon \in (0,t_f)$ is small enough. Then this PDE is the heat equation with Neumann boundary condition. From Lemma \ref{Dave}, it follows that the solution $y$ satisfies $y(\cdot, \epsilon/2)\geq c$ for some positive constant $c$. For each $t \in (\epsilon/3,2\epsilon/3]$, let $\mathbf{v}(\cdot,t)= \frac{\nabla f}{f}$. Then the mild solution of the PDE is given by the semigroup $(\mathcal{T}^B_a(t))_{t \geq 0}$, where $a = 1/f$. From Lemma \ref{cvglma1}, the semigroup $(\mathcal{T}^B_a(t))_{t \geq 0}$ is analytic. Hence, from regularizing properties of analytic semigroups \cite{lunardi2012analytic}[Theorem 2.1.1], it follows that $y(\cdot,\epsilon) \in \mathcal{D}(B^j_a)$ for each $j \in \mathbb{Z}_+$. From Lemma \ref{egdom}, this implies that $\|B_ay(\cdot,2\epsilon/3)\|_{\infty} < \infty$. Due to the product rule (Theorem \ref{prodrule}), Proposition \ref{samedom}, and Lemma \ref{egdom}, this inequality implies that $\|A_ay(\cdot,2\epsilon/3)\|_{\infty} <c$ for some $c>0$.  Let $\mathbf{v}(\cdot,t)=\frac{\nabla y}{y}-\frac{\nabla (a y)}{y}$ for $t \in [2\epsilon/3, \epsilon]$. Then from the last statement of Corollary \ref{CorII} and the fact that the operator $-A_a$ commutes with the semigroup it generates, it follows that $\|\mathcal{M}_aA_ay(\cdot,t)\|_{\infty}=\|\mathcal{M}_a\mathcal{T}_a^A(t-2\epsilon/3)_aA_ay(\cdot,2\epsilon/3)\|_{\infty} <c'$ for some $c'>0$ and for all $t \in (2\epsilon/3, \epsilon]$. Since $a \in W^{1,\infty}(\Omega)$, we can apply the result in Proposition \ref{samedom} and the gradient estimates of the Neumann Laplacian in \cite{grisvard2011elliptic}[Theorem 2.4.2.7] and \cite{maz2009boundedness}[Theorem]. Taken together, all of these observations imply that $\|\nabla y(\cdot,t)\|_{\infty} <k$ for some $k>0$ and all $t \in (2\epsilon/3, \epsilon]$. From Lemma \ref{slblma1}, it also follows that $y$ is uniformly bounded from below, and hence $\mathbf{v}(\cdot,t)$ is essentially bounded for all  $t \in[0,\epsilon]$. Lastly, due to the analyticity of the semigroup $(\mathcal{T}^A_a(t))_{t \geq 0}$, $y(\cdot,\epsilon)  \in \mathcal{D}(A^j_a)$ for each $j \in \mathbb{Z}_+$. Then the result follows from Corollary \ref{prlma2}.
\end{proof}

In the following theorem, we note that system \eqref{eq:Mainsys1} has stronger controllability properties than Theorem \ref{maintheo} describes: this system  is {\it path controllable} if the path is confined to a subset of $L^2(\Omega)$ that is regular enough. This should not be very surprising due to the large dimensionality of its control inputs as compared to the choice of controls in classical PDE control problems, where control inputs are typically localized on a small subset of the interior or boundary of the domain. We restrict the path to the space $W^{2,\infty}(\Omega)$ for simplicity.

\begin{theorem}
Let $\Omega$ be a domain that is either $C^{1,1}$ or convex, and $t_f$ be given. Suppose that $\gamma \in C^1([0,1];W^{2,\infty}(\Omega))$ such that $\gamma(t) \geq c$ for some positive constant  $c$ and for all $t \in [0,t_f]$.  Additionally, suppose that $\int_{\Omega} \gamma(\mathbf{x},t)d\mathbf{x} = 1$ for all $t \in [0,t_f]$. Then there exists $\mathbf{v} \in L^{\infty}([0,1];L^{\infty}(\Omega)^n)$ such that a solution of the PDE \eqref{eq:Mainsys1} satisfies $y(\cdot,t) = \gamma(t)$  for all $t \in [0,t_f]$, provided that $\gamma(0)=y^0$.
\end{theorem}
\begin{proof}
Fix $t \in [0,t_f]$. Consider the solution $\phi(t) \in L^2(\Omega)$ of the Poisson equation in weak form, 
\begin{equation}
\omega_{\mathbf{1}}(\phi(t),\mu) = \big < \frac{\partial \gamma}{\partial t}(t),\mu \big > ~~ \text{for all} ~ \mu \in H^1(\Omega),
\end{equation}
where $\mathbf{1}$ is the function taking value $1$ everywhere on $\Omega$.
Note that since $\int_{\Omega}\gamma(\mathbf{x},t)d\mathbf{x}=1$ for all $t \in [0,t_f]$, we have that $\int_{\Omega} \frac{\partial \gamma}{\partial t} (\mathbf{x},t)d\mathbf{x}  = 0$ for each $t \in [0,t_f]$, and therefore the Poisson equation has a unique solution for each $t \in [0,t_f]$.
Then it follows from \cite{grisvard2011elliptic}[Theorem 2.4.2.7] and Morrey's inequality \cite{leoni2009first}[Theorem 11.34] (when $\Omega$ is a $C^{1,1}$ domain) and  \cite{maz2009boundedness}[Theorem] (when $\Omega$ is convex) that there exists a constant $C$ such that $\|\nabla \phi(t)\|_{\infty} \leq C \|\nabla (\partial \gamma(t)/ \partial t)\|_{2} \leq C \|\nabla(\partial \gamma (t)/ \partial t)\|_{\infty}$ for each $t \in [0,t_f]$. Then setting $\mathbf{v}(\cdot,t) = \frac{\nabla \gamma(t)}{\gamma(t)} - \frac{\nabla \phi(t)}{\gamma(t)}$ for each $t \in [0,t_f]$ gives us the desired controllability result.
\end{proof}

Note that this approach to prove path controllability using the Poisson equation coincides with the approach used by Otto in \cite{otto2001geometry} to define a formal Riemannian metric on the space of probability densities, when the domain $\Omega $ is $\mathbb{R}^n$.

\subsection{Controllability of ODE System \eqref{eq:ctrsys}}

In this section, we investigate the controllability properties of the system \eqref{eq:ctrsys}.

\begin{proposition}
	\label{MK1}
	If the graph $\mathcal{G}=(\mathcal{V},\mathcal{E})$ is not strongly
	connected, then the system \eqref{eq:ctrsys} is not locally controllable.
\end{proposition}

\begin{proof}
	Suppose that $\mathcal{G}=(\mathcal{V},\mathcal{E})$ is not strongly
	connected. Then there exist vertices $v_1,\,v_2\in \mathcal{V}$ such
	that there does not exist a path in $\mathcal{E}$ from $v_2$ to $v_1$.
	Let $\mathcal{V}_1$ be the (non-empty) subset of vertices $v\in \mathcal{V}$ such
	that $v=v_1$ or there exists a path in $\mathcal{E}$
	from $v$ to $v_1$. Analogously,
	let $\mathcal{V}_2$ be the (non-empty) subset of vertices $v\in \mathcal{V}$ such
	that $v=v_2$ or there exists a path in $\mathcal{E}$
	from $v_2$ to $v$. 
	Since there does not exist a path in $\mathcal{E}$ from $v_2$ to $v_1$,
	it is clear that $\mathcal{V}_1$ and $\mathcal{V}_2$ are disjoint.
    By construction, the set $\mathcal{V}_1$ is the largest backward invariant subset of vertices containing $v_1$, and $\mathcal{V}_2$ is the largest forward invariant subset of vertices containing $v_2$.  
Since no edge enters $\mathcal{V}_1$ 
from  $\mathcal{V} -  \mathcal{V}_1$, and no edge leaves $\mathcal{V}_2$, the total
mass in $\mathcal{V}_1$ can only decrease, and the total mass in  $ \mathcal{V}_2$ can only increase.
	Then the output function $\varphi\colon \mathcal{P}(\mathcal{V})\mapsto \mathbb{R}$
	defined by
	\begin{equation}
	\varphi (\boldsymbol{\mu})= \sum_{v\in \mathcal{V}_2}\mu_v- \sum_{v\in \mathcal{V}_1}\mu_v
	\end{equation}
	is nondecreasing along every solution curve of the
	system~\eqref{eq:ctrsys}, which therefore is not locally controllable.
\end{proof}
The following proposition will be used to address Problem \ref{stlcprb}.
\begin{proposition}
	\label{MK2}
	If the graph $\mathcal{G}=(\mathcal{V},\mathcal{E})$ is strongly
	connected, then the system \eqref{eq:ctrsys} is STLC from every point in ${\rm int}(\mathcal{P}(\mathcal{V}))$. 
\end{proposition}

Before proving the proposition, it is helpful to take a closer
look at the relations in the Lie algebra of the control vector fields,
and the corresponding product on the semigroup generated
by their exponentials.

Suppose that $e=(i,j),e'=(k,\ell) \in \mathcal{E}$ are two edges.
If $\{i,j\},\{k,\ell\}\subseteq \mathcal{V}$ are disjoint,
then the control matrices $\mathbf{Q}_{(i,j)}$ and $\mathbf{Q}_{(k,\ell)}$
commute, and hence so do their exponentials.
If $k=j$ and $\ell\neq j$, then $\mathbf{Q}_{(i,j)} \mathbf{Q}_{(j,\ell)}=\mathbf{0}$,
and the commutator evaluates to
\begin{equation}
[\mathbf{Q}_{(j,\ell)},\mathbf{Q}_{(i,j)}]=\mathbf{Q}_{(j,\ell)} \mathbf{Q}_{(i,j)}=\mathbf{Q}_{(i,j)}-\mathbf{Q}_{(i,\ell)}.
\end{equation}
From this, we can conclude that if the graph
is strongly connected, then the Lie algebra spanned
by the control vector fields $f_e\colon \boldsymbol{\mu}\mapsto \mathbf{Q}_e \boldsymbol{\mu}$
 spans the tangent space $T_{\boldsymbol{\mu}} (\mathcal{P}(\mathcal{V}))$ at every point $\boldsymbol{\mu}\in {\rm int}(\mathcal{P}(\mathcal{V}))$. However, since in our case $\mathbf{0}$ is not an interior point
of the convex hull of admissible control values $\mathbf{u}\in [0,\infty)^M$,
classical results on STLC do not apply directly.

For any edge $e=(i,j)\in \mathcal{E}$, the exponential
of the control matrix $\mathbf{Q}_e$ is a stochastic matrix
with entries given by
\begin{equation}
(\exp t\mathbf{Q}_e)_{k\ell} = \left\{ \begin{array}{cl}
1&\mbox{ if } k=\ell\neq S(e) \\
e^{-t}&\mbox{ if } k=\ell=S(e) \\
1-e^{-t}&\mbox{ if } k=T(e) \mbox{ and } \ell=S(e) \\
0&\mbox{ otherwise.}
\end{array} \right.
\end{equation}

Rather than writing out a general formula for the corresponding
product on the group for arbitrary edges $(i,j),(j,\ell)\in \mathcal{E}$,
we only state the product for the special case of
$\mathcal{V}=\{1,2,3\}$ and edges $e=(1,2)$ and $e'=(2,3)$ as an illustration: 
\begin{equation}
e^{t\mathbf{Q}_{e'}}e^{s\mathbf{Q}_e}\!=\!
\left( \begin{array}{ccc}
e^{-s} & 0 & 0\! \\
e^{-t}(1-e^{-s}) & e^{-t} & 0\! \\
(1-e^{-t})(1-e^{-s}) & 1-e^{-t} & 1\!
\end{array} \right).
\end{equation}

\begin{proof} (of Proposition \ref{MK2}).
	Suppose that the graph $\mathcal{G}=(\mathcal{V},\mathcal{E})$
	is strongly connected.
	Fix an arbitrary point $\boldsymbol{\mu}^0\in {\rm int}(\mathcal{P}(\mathcal{V}))$.
	Then there exists $\rho >0$ such that each coordinate
	$\mu^0_i>2\rho$.
	Let $\Delta \boldsymbol{\mu}\in [-\rho/M ,\rho/M ]^M$ be an arbitrary but fixed vector such
	that $\sum_{v\in \mathcal{V}} \Delta \mu_v =0$.
	Let the final time $T>0$ be arbitrary but fixed.
	We explicitly construct a piecewise constant control
	$\mathbf{u} \colon [0,T] \mapsto [0,\infty)^{N_\mathcal{E}}$ that
	steers the system \eqref{eq:ctrsys} from $\boldsymbol{\mu}^0$ at time $0$
	to $\boldsymbol{\mu}^0+\Delta \boldsymbol{\mu}$ at time $T$.
	
	Let $v_0 \in \mathcal{V}$ be arbitrary but fixed.
	As a consequence of Proposition \ref{MK1}, there exists a path $\gamma=(e_1,\ldots, e_s)$ of edges in $\mathcal{E}$
	that connects $v_0=S(e_1)$ back to $T(e_s)=v_0$ and which
	{\em visits} every vertex $v\in \mathcal{V}$ at least once.
	For $1\leq i\leq s$, let $v_i=T(e_i)$.
	Let $\Delta t = T/s$.
	Define the finite sequence $\{\delta_i\}_{i=1}^s\in \{0,1\}^{s}$
	by $\delta_i=1$ if for all $i<j<s$, $v_j\neq v_i$, i.e.
	the edge $e_i\in \gamma$ is the last edge whose source is
	$S(e_i)=v_i$.
	This sequence ensures that a control variation in the direction
	of $\mu_v$ is only taken along the last edge that starts at $v$; in other words, that 
	mass is removed from the vertex $v$ only during the last  time interval when the path visits this vertex.
	Finally, define a finite sequence 
	$\{\sigma_i\}_{i=0}^s\in [-\rho,\rho]^{s}$
	that keeps track of the accumulated control variations, where
	\begin{equation}
	\label{sigmadef}
	\sigma_0=0, ~~~~\sigma_i = \sum_{j=1}^{i} \delta_j\Delta \mu_{v_{j-1}}, ~~~~ 1<i\leq s.
	\end{equation}
	Note that if the path $\gamma$ is a Hamiltonian cycle, then
	$s=M$, $\delta_i = 1$ for all $i=1,...,s$, 
	and $\sigma _i= \sum_{j=1}^{i} \Delta \mu_{v_{j-1}}$,
	which simplifies the formula \eqref{coron} below.
	
	To distinguish between the two cases where $S(e_i)=v_i=v_0$ and $S(e_i)=v_i\neq v_0$, we introduce the vector $\mathbf{y}^0 \in \mathbb{R}^M$ by setting
	$y^0_{v_0}=\mu_{v_0}-\rho$ and $y^0_{v_i}=\mu_{v_i}$ if $v_i\neq v_0$.
	Consider the piecewise constant control
	$\mathbf{u} \colon [0,T] \mapsto [0,\infty)^{N_\mathcal{E}}$ that is defined
	on each interval $t\in [(i-1) \Delta t,i \Delta t)$, $1\leq i <s$, as 
	\begin{equation}
	\label{coron}
	u_{e_i}(t)=
	-\frac{1}{\Delta t}\ln \left( 1-\frac{\rho-\sigma_i}{y^0_{v_i}+\rho-\sigma_{i-1}}
	\right) 
	\end{equation}
	and $u_e(t) \equiv 0$ for all $e\neq e_i$. 
	
	The key idea in this construction is that the much simpler
	control obtained by setting $\Delta \boldsymbol{\mu}=\mathbf{0}$ in definition \eqref{sigmadef} successfully moves
	a mass $\rho>0$ from $v_0$ along the path $\gamma$ and
	back to $v_0$.
	It is critical that the 
	component $u_{e_i}$ of the control be strictly positive on the $i^{th}$ interval of time,
	which enables the application of classical signed control variations to this component on the interval.
	The explicit introduction of the nonzero $\Delta \boldsymbol{\mu}$ then
	allows the following endpoint map to be calculated explicitly:  
	\begin{equation}
	\stackrel{\longleftarrow}{\prod_{1\leq i\leq s}}
	\exp \left(\Delta t u_{e_i} \mathbf{Q}_{e_i} \right) \cdot \boldsymbol{\mu}^0= \boldsymbol{\mu}^0+\Delta \boldsymbol{\mu},
	\end{equation}
	resulting in equation \eqref{coron} for the control.
	
	In the first time interval $[0,\Delta t]$, note that $\sigma_0=0$
	and $y_{v_0}=\mu_{v_0}-\rho$, and hence the 
	only nonzero and non-unity entries of the matrix exponential
	take the simpler forms 
	\begin{align*}
	\left(\exp( \Delta tu_{e_1}\mathbf{Q}_{e_1}) \right)^{v_0v_0}=& ~ e^{-\Delta tu_{e_1}}
	=1-\frac{\rho-\sigma_1}{\mu^0_{v_0}-\rho+\rho-\sigma_0} \nonumber \\
	=& ~ 1-\frac{\rho-\sigma_1}{\mu^0_{v_0}}
	\;,\;\mbox{ and } \nonumber \\
	\left(\exp(\Delta t u_{e_1}\mathbf{Q}_{e_1}) \right)^{v_1v_0} =& ~ 1-e^{-\Delta tu_{e_1}}
	=\frac{\rho-\sigma_1}{\mu^0_{v_0}-\rho+\rho-\sigma_0} \nonumber \\
	=& ~ \frac{\rho-\sigma_1}{\mu^0_{v_0}}.
	\end{align*}
	From this, it readily follows that
	\begin{align*}
	\mu_{v_0}(\Delta t) =& ~ \mu^0_{v_0}-\rho+\sigma_1 = ~ \mu^0_{v_0}-\rho+\delta_1\Delta \mu_{v_0}
	\;,\;\mbox{ and }\;\;\\ \nonumber
	\mu_{v_1}(\Delta t) =& ~ \mu^0_{v_1}+\rho-\sigma_1. \nonumber
	\end{align*}
	
	Now suppose that $1\leq i<s$ is arbitrary but fixed. We consider the following three cases.
	
	\vspace{2mm} 
	
	{\bf Case 1.} 
	$S(e_i)=v_{i}\neq v_0\neq v_{i+1}=T(e_i)$, and hence $y^0_{v_i}=\mu^0_{v_i}$ and
	$y^0_{v_{i+1}}=\mu^0_{v_{i+1}}$, and in particular $i\neq s-1$.
	
	The induction hypothesis is that the following equalities hold:
	\begin{equation}
	\mu_{v_i}(i\Delta t)=\mu^0_{v_i}+\rho-\sigma_{i}
	\;\;\mbox{ and }\;\;
	\mu_{v_{i+1}}(i\Delta t)=\mu^0_{v_{i+1}}. 
	\end{equation}
	The second equality is true because $\delta_i=0$, since we assumed that $i+1\neq s$, which implies that there 
	is an edge $e_j$ with $j=i+1>i$ such that $S(e_{j})=v_{j}$.
	
	For the $(i+1)^{\rm st}$ time interval $[i\Delta t,(i+1)\Delta t)$, note that
	the only nonzero and non-unity entries of the matrix exponential
	are (given that $y^0_{v_i}=\mu^0_{v_i}$): 
	\begin{align*}
	\left(\exp( \Delta tu_{e_{i+1}}\mathbf{Q}_{e_{i+1}}) \right)^{v_iv_i} =& ~ e^{-\Delta tu_{e_{i+1}}} \nonumber \\ 
	=& ~ 1-\frac{\rho-\sigma_{i+1}}{\mu^0_{v_i}+\rho-\sigma_i} 
	\;,\;\mbox{ and } \nonumber  \\ 
	\left(\exp(\Delta t u_{e_{i+1}}\mathbf{Q}_{e_{i+1}}) \right)^{v_{i+1}v_i} =& ~ 1-e^{-\Delta tu_{e_{i+1}}}  \nonumber  \\
	=& ~ \frac{\rho-\sigma_{i+1}}{\mu^0_{v_i}+\rho-\sigma_i}
	\end{align*}
	From this, it readily follows that
	\begin{align*}
	\mu_{v_i}(({i+1})\Delta t) =& ~\mu^0_{v_i}(i\Delta t)-\rho+\sigma_{i+1}
	= ~ \mu^0_{v_i}-\sigma_i +\sigma_{i+1} \nonumber \\
	=& ~ \mu^0_{v_i}+\delta_{i+1}\Delta \mu_{v_{i}}
	\;,\;\mbox{ and }\;\;\\ \nonumber
	\mu_{v_{i+1}}(({i+1})\Delta t) =& ~ \mu^0_{v_{i+1}}(i\Delta t)+\rho-\sigma_{i+1} \nonumber \\
	=& ~ \mu^0_{v_{i+1}}+\rho-\sigma_{i+1}~,
	\end{align*}
	where we have used the recursive definition of $\sigma_j$ in terms of $\delta_j$. 
	\vspace{2mm} 
	
	{\bf Case 2.}
	$S(e_i)\neq v_0= v_{i+1}=T(e_i)$ and $i<s-1$, and hence $y^0_{v_i}=\mu^0_{v_i}$ and
	$y^0_{v_{i+1}}=\mu^0_{v_{i+1}}-\rho$.
	
	The induction hypothesis is that the following equalities hold:
	\begin{equation}
	\mu_{v_i}(i\Delta t)=\mu^0_{v_i}+\rho-\sigma_{i}
	\;\;\mbox{ and }\;\;
	\mu_{v_{i+1}}(i\Delta t)=\mu^0_{v_{i+1}}-\rho. \nonumber
	\end{equation}
	The second equality is true because $\delta_i=0$, since we assumed that $i<s-1$,
	which implies that there is an edge $e_j$ with $j={i+1}>i$ such that 
	$S(e_{j+1})=v_{j+1}$.
	
	For the $(i+1)^{\rm st}$ time interval $[i\Delta t,(i+1)\Delta t)$, note that
	the only nonzero and non-unity entries of the matrix exponential
	are (given that $y^0_{v_i}=\mu^0_{v_i}$):
	\begin{align*}
	\left(\exp( \Delta tu_{e_{i+1}}\mathbf{Q}_{e_{i+1}}) \right)^{v_iv_i} =& ~ e^{-\Delta tu_{e_{i+1}}}  \nonumber \\
	=& ~ 1-\frac{\rho-\sigma_{i+1}}{\mu^0_{v_i}+\rho-\sigma_i}
	\;,\;\mbox{ and } \nonumber \\
	\left(\exp(\Delta t u_{e_{i+1}}\mathbf{Q}_{e_{i+1}}) \right)^{v_{0}v_i} =& ~ 1-e^{-\Delta tu_{e_{i+1}}} \nonumber \\
	= & ~ \frac{\rho-\sigma_{i+1}}{\mu^0_{v_i}+\rho-\sigma_i} \nonumber 
	\end{align*}
	From this, it readily follows that
	\begin{align*}
	\mu_{v_i}(({i+1})\Delta t) =& ~ \mu^0_{v_i}(i\Delta t)-\rho+\sigma_{i+1} \nonumber \\
	=& ~ \mu^0_{v_i}+\sigma_i +\sigma_{i+1}=   \mu^0_{v_i}+\delta_{i+1}\Delta \mu_{v_{i}} 
	\;,\;\mbox{ and } \nonumber \\
	\mu_{v_{0}}(({i+1})\Delta t) =& ~ \mu^0_{v_{0}}(i\Delta t)-\rho+\rho-\sigma_{i+1} \nonumber \\
	=& ~ \mu^0_{v_{0}}+\sigma_{i+1}~, \nonumber 
	\end{align*}
	where we have used the recursive definition of $\sigma_j$ in terms of $\delta_j$. 
	
	\vspace{2mm} 
	
	{\bf Case 3.}
	$S(e_i)=v_i= v_0\neq T(e_i)$ and $i<s-1$, and hence $y^0_{v_i}=\mu^0_{v_i}-\rho$ and
	$y^0_{v_{i+1}}=\mu^0_{v_{i+1}}$.
	
	The induction hypothesis is that the following equalities hold:
	\begin{equation}
	\mu_{v_0}(i\Delta t)=\mu^0_{v_0}-\sigma_{i}
	\;\;\mbox{ and }\;\;
	\mu_{v_{i+1}}(i\Delta t)=\mu^0_{v_{i+1}}.
	\end{equation}
	The second equality is true because $\delta_i=0$, since we assumed that $i<s-1$,
	which implies that there is an edge $e_j=e_{i+1}$ such that $S(e_{j+1})=v_{j+1}$.
	
	For the $(i+1)^{\rm st}$ time interval $[i\Delta t,(i+1)\Delta t)$, note that
	the only nonzero and non-unity entries of the matrix exponential
	are (given that $y^0_{v_i}=\mu^0_{v_i}$): 
	\begin{align*}
	\left(\exp( \Delta tu_{e_{i+1}}\mathbf{Q}_{e_{i+1}}) \right)^{v_0v_0} = & ~ e^{-\Delta tu_{e_{i+1}}}  \nonumber \\
	=& ~ 1-\frac{\rho-\sigma_1}{\mu^0_{v_0}-\rho+\rho-\sigma_0} \nonumber \\
	=& ~ 1-\frac{\rho-\sigma_1}{\mu^0_{v_0}}
	\;,\;\mbox{ and } \nonumber \\ 
	\left(\exp(\Delta t u_{e_{i+1}}\mathbf{Q}_{e_{i+1}}) \right)^{v_{i+1}v_0} =& ~ 1-e^{-\Delta tu_{e_{i+1}}} \nonumber \\
	=& ~ \frac{\rho-\sigma_1}{\mu^0_{v_0}-\rho+\rho-\sigma_0} \nonumber \\
	=& ~\frac{\rho-\sigma_1}{\mu^0_{v_0}} \nonumber
	\end{align*}
	From this, it readily follows that
	\begin{align*}
	\mu_{v_0}(({i+1})\Delta t) = &~\mu^0_{v_0}(i\Delta t)-\rho+\sigma_{i+1} \nonumber \\
	= &~\mu^0_{v_0}-\rho +\sigma_i +\sigma_{i+1} \nonumber \\
	= &~\mu^0_{v_0}-\rho+\delta_{i+1}\Delta \mu_{v_{i+1}}
	\;,\;\mbox{ and }\;\;  \nonumber \\
	\mu_{v_{i+1}}(({i+1})\Delta t) =&~\mu^0_{v_{i+1}}(i\Delta t)+\rho-\sigma_{i+1}  \nonumber \\
	= &~ \mu^0_{v_{i+1}}+\rho-\sigma_{i+1}~,  \nonumber 
	\end{align*}
	where we have used the recursive definition of $\sigma_j$ in terms of $\delta_j$.
	By the definition of $\delta_j$, it is clear that after
	the last edge starting
	at $v_i$, $\mu_{v_j}(t)$ stays constant with value
	$\mu_{v_j}(t)=\mu^0_j+\Delta \mu_i$ for all $(j+1)\Delta t \leq t \leq T$.
	
	In the final step on the time interval $[T-\Delta t, T]$ corresponding to 
	the edge $e_s=(v_{s-1},v_s)=(v_{s-1},v_0)$, the initial conditions are
	\begin{align*}
	\mu_{v_{s-1}}((s-1)\Delta t) =& ~ \mu^0_{v_{s-1}}+\rho-\sigma_{s-1}
	\;,\;\mbox{ and }\;\; \nonumber \\
	\mu_{v_{0}}((s-1)\Delta t) =& ~ \mu^0_{v_0}-\rho-\delta_s\Delta \mu_{v_0}  \nonumber \\               
	=& ~ \mu^0_{v_0}-\rho - \Delta \mu_{v_0}. \nonumber              
	\end{align*}
	In the second equation, $\delta_s=1$ because there is no  subsequent edge in the path $\gamma$ that starts from $v_0$.
	
	For this $s^{\rm th}$ time interval $[T-\Delta t,T]$, note that
	the only nonzero and non-unity entries of the matrix exponential
	are (given that $y^0_{v_{s-1}}=\mu^0_{v_{s-1}}$): 
	\begin{align*}
	\left(\exp( \Delta tu_{e_s}\mathbf{Q}_{e_s}) \right)^{v_{s-1}v_{s-1}} =& ~ e^{-\Delta tu_{e_s}} \nonumber \\
	=& ~1-\frac{\rho-\sigma_s}{\mu^0_{v_{s-1}}+\rho-\sigma_{s-1}}
	\;,\;\mbox{ and } \nonumber \\
	\left(\exp(\Delta t u_{e_s}\mathbf{Q}_{e_s}) \right)^{v_{0}v_{s-1}} =& ~ 1-e^{-\Delta tu_{e_s}} \nonumber \\
	=& ~ \frac{\rho-\sigma_s}{\mu^0_{v_i}+\rho-\sigma_{s-1}} \nonumber 
	\end{align*}
	From this, it readily follows that
	\begin{align*}
	\mu_{v_{s-1}}(T) =& ~ \mu^0_{v_{s-1}}((s-1)\Delta t)-\rho-\sigma_s \nonumber \\
	=& ~ \mu^0_{v_{s-1}}+\sigma_{s-1} -\sigma_{s} \nonumber \\
	=& ~ \mu^0_{v_{s-1}}+\Delta \mu_{v_{s-1}}
	\;,\;\mbox{ and }\;\; \nonumber \\
	\mu_{v_{0}}(T) =& ~ \mu^0_{v_{0}}(({s-1})\Delta t)-\rho+\rho-\sigma_s \nonumber \\
	=& ~ \mu^0_{v_{0}}+\sigma_s 
	= ~ \mu^0_{v_{0}}-\sum_{v\in \mathcal{V}\setminus v_0} \Delta \mu_{v} 
	= ~ \mu^0_{v_{0}}+ \Delta \mu_{v_0}. \nonumber 
	\end{align*}
	Here, we have used the recursive definition of $\sigma_j$ in terms of $\delta_j$
	and the fact that $\sum_{v\in \mathcal{V}} \Delta \mu_{v}=0$.
\end{proof}



\begin{theorem}
	\label{MK0}
	Let $T>0$. If the graph
	$\mathcal{G}=(\mathcal{V},\mathcal{E})$ is strongly
	connected, then the system \eqref{eq:ctrsys} is globally
	controllable within $T$ from every point in the interior of the simplex $\mathcal{P}(\mathcal{V})$.
\end{theorem}

\begin{proof}
	Suppose that $\boldsymbol{\mu}^0, \boldsymbol{\mu}^T\in {\rm int}(\mathcal{P}(\mathcal{V}))$.
	Let $\rho=\frac{1}{2} \min \{  \mu^0_v, \mu^T_v \colon v\in \mathcal{V}\}$,
	$L=\| \boldsymbol{\mu}^T- \boldsymbol{\mu}^0\|_1$, and $N={\rm ceil}(L/\rho)$.
	Partition the straight-line segment from $\boldsymbol{\mu}^0$ to $\boldsymbol{\mu}^T$ into $N$ segments,
	e.g. with endpoints $\mathbf{y}^k= \boldsymbol{\mu}^0+\frac{k}{N}(\boldsymbol{\mu}^T- \boldsymbol{\mu}^0)$ for $0\leq k \leq N$.
	Using Proposition \ref{MK2}, there exist controls $\mathbf{u}^k\colon [\frac{kT}{N},\frac{(k+1)T}{N}]\mapsto [0,\infty)^{N_\mathcal{E}}$
	that successively steer the system from $\boldsymbol{\mu}^k$ to $\boldsymbol{\mu}^{k+1}$.
	Thus, the concatenation of these controls steers the system
	from $\boldsymbol{\mu}^0$ to $\boldsymbol{\mu}^T$ in time $T$ using piecewise constant
	controls that take values only on the axes of $\mathbb{R}_{+}^{M}$.
\end{proof}

\subsection{Controllability of the PDE System \eqref{eq:Mainsys3}}

In this section, we will consider the controllability problem described in Problem \ref{ADRctrl2}.
We define some new notation that will be needed in this section and the following one. These definitions will be used to construct solutions of the system of PDEs \eqref{eq:Mainsys2} and hence enable the controllability and stability analysis.
Let $\mathbf{a} = [a_1 ~ a_2 ~ ... ~ a_{N}]^T \in \mathbf{L}^{\infty}(\Omega) = Z_1 \times ... \times Z_N$,  where $a_i \in L^{\infty}(\Omega)$ and $Z_i = L^{\infty}(\Omega)$ for each $i \in \mathcal{V}$. If $c > 0$, then we write $\mathbf{a} \geq  c$ to denote that $a_i \geq c $ for each $i \in \mathcal{V}$. We will  assume throughout that this condition is satisfied by $\mathbf{a}$ for some positive constant $c$. We consider the operator $\boldsymbol{\mathcal{B}}_{\mathbf{a}}: \mathcal{D}(\boldsymbol{\mathcal{B}}_{\mathbf{a}}) \rightarrow \mathbf{L}^2_{\mathbf{a}}(\Omega)$, where $\mathbf{L}^2_{\mathbf{a}}(\Omega) = L^2_{a_1}(\Omega) \times ... \times L^2_{a_{N}}(\Omega)$ is equipped with the norm $\|\cdot\|_{\mathbf{a}}$, defined as $\|\mathbf{u}\|_{\mathbf{a}} = (\sum_{i=1}^{N} \|u_i\|^2_{a})^{1/2}$ for each $\mathbf{u} = [u_1 ~ ...~ u_{N}]^T \in \mathbf{L}^2_{\mathbf{a}}(\Omega)$, and $\mathcal{D}(\boldsymbol{\mathcal{B}}_{\mathbf{a}})=\mathcal{D}(B_{a_1}) \times \mathcal{D}(B_{a_2}) \times...\times \mathcal{D}(B_{a_N})$. The operator $\boldsymbol{\mathcal{B}}_{\mathbf{a}}$ is defined by $\boldsymbol{\mathcal{B}}_\mathbf{a}\mathbf{v} = [B_{a_1}v_1 ~ B_{a_2}v_2 ~ ... ~ B_{a_N}v_N]^T$ for each $\mathbf{v} =  [v_1 ~ ...~ v_N]^T \in \mathcal{D}(\boldsymbol{\mathcal{B}}_{\mathbf{a}})$. Corresponding to each matrix $\mathbf{Q}_e$, we associate a bounded operator $\boldsymbol{\mathcal{Q}}_e$ on the space $\mathbf{L}^2_{\mathbf{a}}(\Omega)$ given by $\boldsymbol{(\mathcal{Q}}_e \mathbf{y})(\mathbf{x}) = \mathbf{Q}_e \mathbf{y}(\mathbf{x}) $ for each $\mathbf{y} = [y_1 ~  ... ~ y_{N}]^T \in \mathbf{L}^2_{\mathbf{a}}(\Omega)$ and a.e. $\mathbf{x} \in \Omega$. Let $b \in \mathbf{L}^{\infty}(\Omega)$. $\boldsymbol{\mathcal{M}}_{\mathbf{b}}$ will denote the multiplication operator defined by $\boldsymbol{\mathcal{M}}_\mathbf{b}\mathbf{v}= [\mathcal{M}_{b_1}v_1 ~ \mathcal{M}_{b_2}v_2 ~ ... ~ \mathcal{M}_{b_N}v_N]^T$ for each $\mathbf{v} \in \mathbf{L}^2(\Omega) =  \mathbf{L}^2_{\mathbf{a}}(\Omega)$. For a function $K_e \in L^{\infty}(\Omega)$, $K_e \boldsymbol{\mathcal{Q}}_e$ will denote the product operator $\boldsymbol{\mathcal{M}}_{\mathbf{b}}\boldsymbol{\mathcal{Q}}_e$, where $\boldsymbol{\mathcal{M}}_{\mathbf{b}}$ is the multiplication operator corresponding to the function $\mathbf{b}\in \mathbf{L}^{\infty}(\Omega)$ defined by setting $b_i = K_e$ for each $i \in \mathcal{V}$.

Using these definitions, we construct some semigroups on the space $\mathbf{L}_{\mathbf{a}}^2(\Omega)$ that will be used to address Problem \ref{ADRctrl2}.

\begin{lemma}
\label{syspos}
Let $\lbrace K_e \rbrace_{e \in \mathcal{E}}$ be a set of non-negative functions in $L^{\infty}(\Omega)$. Suppose $\mathbf{b} \in \mathbf{L}^{\infty}(\Omega)$ such that $b_i =D_i\mathbf{1}$ is a positive constant function for each $i \in \mathcal{V}$. Then the operator $-\boldsymbol{\mathcal{M}}_{\mathbf{b}}\boldsymbol{\mathcal{B}}_{\mathbf{a}} + \sum_{e \in \mathcal{E}}K_e\boldsymbol{\mathcal{Q}}_e$ generates a semigroup of operators $(\boldsymbol{\mathcal{S}}(t))_{t \geq 0}$ on $\mathbf{L}^2_{\mathbf{a}}(\Omega)$. Moreover, the semigroup is positive and {\it mass-conserving}, i.e., if $\mathbf{y}^0 \in \mathbf{L}^2_\mathbf{a}(\Omega)$ is real-valued, then $\sum_{i \in \mathcal{V}}\int_{\Omega} (\boldsymbol{\mathcal{S}}(t)\mathbf{y}^0)_i(\mathbf{x})d\mathbf{x} = \sum_{i \in \mathcal{V}}\int_{\Omega} {y}_i(\mathbf{x})d\mathbf{x}$ for all $t \geq 0$.

Additionally, if $a_i \in W^{1,\infty}(\Omega)$, then $\boldsymbol{\mathcal{S}}(t)\mathbf{y}^0$ is the unique mild solution of the system \eqref{eq:Mainsys2} with $f_i = 1/a_i$,  $\mathbf{v}_i(\cdot,t) = D_i\nabla f_i/f_i$, and $u_e(t)=K_e$ for all $i \in \mathcal{V}$, all $e \in \mathcal{E}$, and all $t \in [0,t_f]$.
\end{lemma}
\begin{proof}
The generation of the semigroup $(\boldsymbol{\mathcal{S}}(t))_{t \geq 0}$  follows from the fact that $-\boldsymbol{\mathcal{M}}_{\mathbf{b}}\boldsymbol{\mathcal{B}}_{\mathbf{a}} + \sum_{e \in \mathcal{E}}K_e\boldsymbol{\mathcal{Q}}_e$ is a bounded perturbation of the operator $-\boldsymbol{\mathcal{M}}_{\mathbf{b}}\boldsymbol{\mathcal{B}}_{\mathbf{a}}$. The positivity preserving property of the semigroup can be demonstrated as follows using the {\it Lie-Trotter product formula} \cite{engel2000one}[Corollary III.5.8]. Let $(\boldsymbol{\mathcal{U}}(t))_{t \geq 0}$ be the semigroup generated by the operator $\sum_{e \in \mathcal{E}}K_e\boldsymbol{\mathcal{Q}}_e$. In fact, the semigroup can be explicitly represented as $\boldsymbol{\mathcal{U}}(t) = e^{\sum_{e \in \mathcal{E}}K_e\boldsymbol{\mathcal{Q}}_e t}$ for each $t \geq 0$. Moreover, $(e^{\sum_{e \in \mathcal{E}}K_e\boldsymbol{\mathcal{Q}}_e t}\mathbf{y}^0)(\mathbf{x}) = e^{\sum_{e \in \mathcal{E}}K_e(\mathbf{x})\mathbf{Q}_e t}\mathbf{y}^0(\mathbf{x})$ for each $\mathbf{y}^0 \in \mathbf{L}^2_\mathbf{a}(\Omega)$ and a.e. $\mathbf{x} \in \Omega$. The semigroup $(\boldsymbol{\mathcal{U}}(t))_{t \geq 0}$ is positivity preserving since each matrix $\mathbf{Q}_e$ has positive off-diagonal entries. From Corollary \ref{CorII}, we also note that the semigroup $(\boldsymbol{\mathcal{V}}(t))_{t \geq 0}$ generated by the operator $-\boldsymbol{\mathcal{M}}_{\mathbf{b}}\boldsymbol{\mathcal{B}}_{\mathbf{a}}$ is positivity preserving. Moreover, since $\sum_{e \in \mathcal{E}}K_e\boldsymbol{\mathcal{Q}_e}$ is a bounded operator, there exists $w \in \mathbb{R}$ such that $\|\boldsymbol{\mathcal{U}}(t)\|_{op} \leq M e^{w t}$ for some positive constant $M$ for all $t \geq 0$.
The semigroup $(\boldsymbol{\mathcal{V}}(t))_{t \geq 0}$ is {\it contractive}, i.e., $\|\boldsymbol{\mathcal{V}}(t)\|_{op} \leq 1$ for all $t \geq 0$. Hence, it follows from the Lie-Trotter product formula that $\boldsymbol{\mathcal{S}}(t)\mathbf{y}^0 = \lim_{n \rightarrow \infty}[\boldsymbol{\mathcal{V}}(t/n)\boldsymbol{\mathcal{U}}(t/n)]^n\mathbf{y}^0 $ for all $t \geq 0$. Therefore, the semigroup $(\boldsymbol{\mathcal{S}}(t))_{t \geq 0}$ is positivity preserving. Through another application of the Lie-Trotter product formula, it follows that $\sum_{i \in \mathcal{V}}\int_{\Omega} (\boldsymbol{\mathcal{S}}(t)\mathbf{y}^0)_i(\mathbf{x})d\mathbf{x} = \sum_{i \in \mathcal{V}}\int_{\Omega} y^0_i(\mathbf{x})d\mathbf{x}$ for all $t \geq 0$.
\end{proof}

In the following lemma, we identify a relation between solutions of the system of PDEs \eqref{eq:Mainsys2} and solutions of the ODE \eqref{eq:ctrsys}. 

\begin{lemma}
\label{eqvsol}
Let $\lbrace q_e \rbrace_{e \in \mathcal{E}}$ be a set of non-negative constants. Suppose $\mathbf{b} \in \mathbf{L^{\infty}}(\Omega)$ such that $b_i =D_i\mathbf{1}$ is a positive constant function for each $i \in \mathcal{V}$. Let $(\boldsymbol{\mathcal{S}}(t))_{t \geq 0}$ be the semigroup generated by the operator $-\boldsymbol{\mathcal{M}}_{\mathbf{b}}\boldsymbol{\mathcal{B}}_{\mathbf{a}} + \sum_{e \in \mathcal{E}}q_e\boldsymbol{\mathcal{Q}}_e$. Additionally, assume that $\mathbf{y}^0 \in \mathcal{D}(-\boldsymbol{\mathcal{M}}_{\mathbf{b}}\boldsymbol{\mathcal{B}}_{\mathbf{a}})$ is real-valued and that $\mu^0_i = \int_{\Omega} y^0_i(\mathbf{x})d\mathbf{x}$ for each $i \in \mathcal{V}$. Then the solution of the system \eqref{eq:ctrsys} satisfies $\mu_i(t) =\int_{\Omega}(\boldsymbol{\mathcal{S}}(t)\mathbf{y}^0)_i(\mathbf{x})d \mathbf{x}$ for each $t \geq 0$ and each $i \in \mathcal{V}$.
\end{lemma}
\begin{proof}
Let $\mathbf{y}(\cdot,t) = \boldsymbol{\mathcal{S}}(t)\mathbf{y}^0$ for each $t \geq 0$.
Then the result follows by noting that
\begin{align*}
\frac{d}{dt}\int_{\Omega} & y_i(\mathbf{x},t)d\mathbf{x} \\ \nonumber
 =& ~ \int_{\Omega} D_iB_{a_i} y_{i}(\mathbf{x},t)d\mathbf{x} +  \sum_{j = 1}^{N}\sum_{e \in \mathcal{E}} \int_{\Omega} q_e Q_e^{ij}y_i(\mathbf{x},t)d\mathbf{x} \\ \nonumber
=& ~ D_i\sigma_{a_i}(y_i,1/a_i) + \sum_{j = 1}^{N}\sum_{e \in \mathcal{E}} q_eQ_e^{ij} \int_{\Omega} y_i(\mathbf{x},t)d\mathbf{x} \\ \nonumber
=& ~ \sum_{j = 1}^{N}\sum_{e \in \mathcal{E}} q_eQ_e^{ij} \int_{\Omega} y_i(\mathbf{x},t)d\mathbf{x}
\end{align*}
for all $t \geq 0$.
\end{proof}
The lemma above allows us to apply the results of Theorems \ref{maintheo} and  \ref{MK0} to prove the following controllability result, which addresses Problem \ref{ADRctrl2}.
\begin{theorem}
Let $\Omega$ be a domain that is $C^{1,1}$ or convex and that satisfies the chain condition. Let $t_f>0$. Let $f_i \in W^{1,\infty}(\Omega)$ for each $i \in \mathcal{V}$ such that $f_i \geq c$ for some positive constant $c$. Suppose $\mathbf{y}^0 \in \mathbf{L}^2(\Omega) $ such that $\mathbf{y}^0 \geq 0$ and $\sum_{i}\int_{\Omega}f_i(\mathbf{x})d\mathbf{x} = \sum_{i}\int_{\Omega}y^0_i(\mathbf{x})d\mathbf{x}$. Then there exist control parameters $\lbrace v_i \rbrace_{i \in V}$ in $L^\infty([0,t_f];L^\infty(\Omega)^n)$ and $\lbrace u_e \rbrace_{e \in \mathcal{E}}$ in $L^\infty([0,t_f])$, where each $u_e$ is non-negative, such that the unique mild solution of the system \eqref{eq:Mainsys2} satisfies $y_i(\cdot,t_f) = f_i $ for each $i \in \mathcal{V}$.
\end{theorem}
\begin{proof}
Let $\mathbf{v}_i(\cdot,t) = \mathbf{0}$ for each $t \in [0,t_f/2]$ and for each $i \in \mathcal{V}$. Then from Theorem \ref{MK0} and Lemmas  \ref{syspos} and \ref{eqvsol}, it follows that there exist piecewise constant parameters $u_e:[0,t_f/2] \rightarrow \mathbb{R}_{+} $ such that the mild solution of the PDE \eqref{eq:Mainsys2} satisfies $\int_{\Omega}y_i(\mathbf{x},t_f/2)d\mathbf{x} = \int_{\Omega}f_i(\mathbf{x})d\mathbf{x}$ for each $i \in \mathcal{V}$. Then the result follows by extending the function $u_e$ to the domain $[0,t_f]$ by defining $u_e(t) = 0$ for $t \in (t_f/2,t_f]$ and by defining $\mathbf{v}_i(\cdot,t)$ 
for $t \in (t_f/2,t_f]$ as in the proof of Theorem \ref{maintheo}.
\end{proof}

\section{Stabilizing Desired Stationary Distributions}
\label{stabsec}
In this section, we address Problem \ref{stabprb}. We first briefly review the notion of {\it irreducibility} of a {\it positive operator} \cite{meyer2012banach}, which will be used extensively in the theorems in this section. Let  $\mathcal{P}$ be a positive operator on the Hilbert space $X = L^2_{a}(\Omega)$ (or $\mathbf{L}^2_\mathbf{a}(\Omega)$) for some $a \in L^{\infty}(\Omega)$ (or $\mathbf{a} \in \mathbf{L}^{\infty}(\Omega)$), i.e., a linear bounded operator that maps real-valued non-negative elements of $X$ to real-valued non-negative elements of $X$. Let $\tilde{\Omega} \subset \Omega$ (or $\tilde{\Omega} \subset \Omega^N$) be a measurable subset. Consider the set $\mathcal{I}_{\tilde{\Omega}}$ defined by 
$
\mathcal{I}_{\tilde{\Omega}} = \big  \lbrace f \in X : \tilde{\Omega} \subset \lbrace \mathbf{x} \in \Omega:   f(\mathbf{x}) = 0\rbrace \big \rbrace 
$.
$\mathcal{P}$ will be called {\it irreducible} if the only measurable sets $\tilde{\Omega} \subset \Omega$ for which the set $\mathcal{I}_{\tilde{\Omega}}$ is invariant under $\mathcal{P}$ are $\tilde{\Omega} = \Omega $ (or $\Omega^N$) and $\tilde{\Omega} = \varnothing$, the null set. A positive semigroup of operators $(\mathcal{T}(t))_{t \geq 0}$ on $X$ will be called irreducible if $\mathcal{T}(t)$ is an irreducible operator for every $t > 0$. Suppose that $A$ is the generator of the semigroup $(\mathcal{T}(t))_{t \geq 0}$ and $s(A) := {\rm sup} \lbrace {\rm Re}(\lambda) : \lambda \in  {\rm spec}(A) \rbrace$. Then $(\mathcal{T}(t))_{t \geq 0}$ being irreducible is equivalent to $(\lambda \mathbb{I} - A)^{-1}$ mapping real-valued non-negative elements of $X$ to strictly positive elements of $X$ for every $\lambda > s(A)$ \cite{arendt2006one}[Definition C-III.3.1]. 
Note that the definitions in the cited reference are stated in a general framework of {\it Banach lattices}, for which $(\mathcal{T}(t))_{t \geq 0}$ being irreducible is equivalent to $(\lambda \mathbb{I} - A)^{-1}$ mapping positive elements of $X$ to {\it quasi-interior} elements of $X$. However, for the spaces that we consider, quasi-interior elements are the same as functions that are positive almost everywhere on their domain of definition.

\begin{theorem}
\label{spelc1}
Let $\lbrace q_e \rbrace _{e \in \mathcal{E}}$ be a set of non-negative constants. Then ${\rm spec}(\sum_{e \in \mathcal{E}}q_e \mathbf{Q}_e) \subset \bar{\mathbb{C}}_-$.
\end{theorem}
\begin{proof}
This follows from \cite{minc1988nonnegative}[Theorem II.1.1] by noting that all the elements of the matrix $\mathbf{G}_\lambda = \lambda \mathbb{I}+\sum_{e \in \mathcal{E}}q_e \mathbf{Q}_e$ are non-negative for $\lambda >0$ large enough.
\end{proof}

\begin{theorem}
\label{spelc2}
Let $\lbrace q_e \rbrace _{e \in \mathcal{E}}$ be a set of non-negative constants. Let $\mathbf{a} \in \mathbf{L}^{\infty}(\Omega)$ such that $\mathbf{a} \geq c$ for some positive constant $c$. Suppose $\mathbf{b} \in \mathbf{L^{\infty}}(\Omega)$ such that $b_i =D_i\mathbf{1}$ is a positive constant function for each $i \in \mathcal{V}$. Then ${\rm spec}(-\boldsymbol{\mathcal{M}}_{\mathbf{b}}\boldsymbol{\mathcal{B}}_{\mathbf{a}}+ \sum_{e \in \mathcal{E}} q_ea_{S(e)} \boldsymbol{\mathcal{Q}}_e) \subset \bar{\mathbb{C}}_-$. 
\end{theorem}
\begin{proof}
Let $\boldsymbol{\mathcal{W}} = -\boldsymbol{\mathcal{M}}_{\mathbf{b}}\boldsymbol{\mathcal{B}}_{\mathbf{a}}+ \sum_{e \in \mathcal{E}} q_ea_{S(e)} \boldsymbol{\mathcal{Q}}_e$. Let $\lambda \in \mathbb{C}  \backslash  {\rm spec}(-\boldsymbol{\mathcal{M}}_{\mathbf{b}}\boldsymbol{\mathcal{B}}_{\mathbf{a}}+ \sum_{e \in \mathcal{E}} q_ea_{S(e)} \boldsymbol{\mathcal{Q}}_e$) be real and large enough such that $(\lambda\mathbb{I} - \boldsymbol{\mathcal{W}})^{-1}$ is a positive operator, i.e., $(\lambda\mathbb{I} - \boldsymbol{\mathcal{W}})^{-1}f \geq 0$ whenever $f \geq 0$. Such a $\lambda$ necessarily exists because the semigroup $(\boldsymbol{\mathcal{U}}(t))_{t \geq 0}$ generated by the operator $\boldsymbol{\mathcal{W}}$ is positivity preserving from Lemma \ref{syspos}. Hence, the existence of $\lambda$ follows from the resolvent formula $(\lambda\mathbb{I} - \boldsymbol{\mathcal{W}})^{-1} = \int_0^{\infty} e^{-\lambda t}\boldsymbol{\mathcal{U}}(t)dt$ when $\lambda$ is greater than the growth bound of the semigroup $(\boldsymbol{\mathcal{U}}(t))_{t \geq 0}$, which is equal to the spectral growth bound $s(\boldsymbol{\mathcal{W}})$ of $\boldsymbol{\mathcal{W}}$ since $(\boldsymbol{\mathcal{U}}(t))_{t \geq 0}$ is analytic \cite{engel2000one}[Theorem II.1.10]. Let $R_{\lambda} =  (\lambda \mathbb{I} - \boldsymbol{\mathcal{W}})^{-1}$. The operator $-\boldsymbol{\mathcal{M}}_{\mathbf{b}}\boldsymbol{\mathcal{B}}_{\mathbf{a}}$ has a compact resolvent since $H^1_{a_i}(\Omega)$ is compactly embedded in $L^2_{a_i}(\Omega)$ for each $i \in \mathcal{V}$ \cite{schmudgen2012unbounded}[Proposition 10.6]. The operator $R_{\lambda}$ is compact and positivity preserving since $\boldsymbol{\mathcal{W}}$ is a bounded perturbation of $-\boldsymbol{\mathcal{M}}_{\mathbf{b}}\boldsymbol{\mathcal{B}}_{\mathbf{a}}$. Additionally, the spectral radius of $R_\lambda$ is positive since $0$ is an eigenvalue of $\boldsymbol{\mathcal{W}}$ (and hence $\frac{1}{\lambda}$ is an eigenvalue of $R_{\lambda}$). Therefore, from the {\it Krein-Rutman theorem} \cite{meyer2012banach}[Theorem 4.1.4], it follows that if $r$ is the spectral radius of the operator $R_\lambda$, then there exists a positive nonzero element $\mathbf{h} \in \mathbf{L}^2_{\mathbf{a}}(\Omega)$ such that $r\mathbf{h} - R_\lambda \mathbf{h} = \mathbf{0} $. Then it follows that $\mathbf{h} \in \mathcal{D}(\boldsymbol{\mathcal{W}})$ and $(\lambda -\frac{1}{r}) \mathbf{h} - \boldsymbol{\mathcal{W}} \mathbf{h} = \mathbf{0}$. For the sake of contradiction,
suppose that $\lambda > \frac{1}{r}$.  Then we have that 
\begin{align*}
\alpha  \int_{\Omega} h_i(\mathbf{x})d\mathbf{x}  ~+~& \int_{\Omega}D_i(B_{a_i}h_i)(\mathbf{x})d\mathbf{x} \nonumber \\ 
~-~&\sum_{e \in \mathcal{E}} \sum_{j=1}^N \int_{\Omega}q_ea_{S(e)}(\mathbf{x}) Q^{ij}_eh_j(\mathbf{x})d\mathbf{x} = 0
\end{align*}
for each $i \in \mathcal{V}$, where $\alpha = \lambda -\frac{1}{r}$.
This implies that
\begin{equation}
\alpha  \int_{\Omega} h_i(\mathbf{x})d\mathbf{x} -\sum_{e \in \mathcal{E}} \sum_{j=1}^{N} \int_{\Omega} q_ea_{S(e)}(\mathbf{x})Q^{ij}_eh_j(\mathbf{x})d\mathbf{x} = 0 \nonumber
\end{equation}
for each $i \in \mathcal{V}$. But this implies that the matrix $\sum_{e \in \mathcal{E}}q_e k_{S(e)} \mathbf{Q}_e$, where the constants $\lbrace k_i \rbrace_{i \in \mathcal{V}}$ are such that
\begin{equation}
\int_{\Omega} a_{i}(\mathbf{x})h_i(\mathbf{x})d\mathbf{x} = k_i\int_{\Omega} h_i(\mathbf{x})d\mathbf{x},  \nonumber
\end{equation}
 has a positive eigenvalue $\alpha$. This contradicts Theorem \ref{spelc1}, since ${\rm spec}(\sum_{e \in \mathcal{E}}q_e k_{S(e)} \mathbf{Q}_e) \subset  \bar{\mathbb{C}}_-$.
\end{proof}

\begin{proposition}
\label{eq:stabsc}
Let $\mathcal{G}$ be strongly connected. Let $\mathbf{f} \in \mathbf{L}^{\infty}(\Omega)$ be such that $\mathbf{f} \geq c $ for some positive constant $c$. Let $\mathbf{b} \in \mathbf{L}^{\infty}(\Omega)$ such that $b_i =D_i\mathbf{1}$ is a positive constant function for each $i \in \mathcal{V}$.
Suppose $\mathbf{y}^0 \in \mathbf{L}^2(\Omega) $ such that $\mathbf{y}^0 \geq 0$ and $\sum_{i}\int_{\Omega}f_i(\mathbf{x})d\mathbf{x} = \sum_{i}\int_{\Omega}y^0_i(\mathbf{x})d\mathbf{x} = 1$. Let $\mathbf{a}\in \mathbf{L}^{\infty}(\Omega)$ be such that $a_i = 1/f_i$ for each $i \in \mathcal{V}$. Then there exist positive parameters $\{q_e\}_{e \in \mathcal{E}}$ such that, if $(\boldsymbol{\mathcal{V}}(t))_{t \geq 0}$ is the semigroup generated by the operator $-\boldsymbol{\mathcal{M}}_{\mathbf{b}}\boldsymbol{\mathcal{B}}_{\mathbf{a}} + \sum_{e \in \mathcal{E}} q_ea_{S(e)}\boldsymbol{\mathcal{Q}}_e$, then we have 
\begin{equation}
\|\boldsymbol{\mathcal{V}}(t)\mathbf{y}^0- \mathbf{f}\|_2 ~\leq~ Me^{-\lambda t}
\end{equation}
for some positive constants $M$ and $\lambda$ and all $t \geq 0$.
\end{proposition}
\begin{proof}
Since the graph $\mathcal{G}$ is assumed to be strongly connected, from \cite{elamvaz2016lin}[Theorem IV.5] (see \cite{biswal2017mean} for proof) we know that there exist positive parameters $\{q_e\}_{e \in \mathcal{E}} $ such that, if $u_e(t) = q_e$ for all $e \in \mathcal{E}$ and all $t \geq 0$, then the solution $\boldsymbol{\mu}(t)$ of the system \eqref{eq:ctrsys} satisfies
\begin{equation}
\|\boldsymbol{\mu}(t)- \boldsymbol{\mu}^{eq}\|_2 ~\leq~ M_1 e^{-\lambda_1 t}
\end{equation}
for some positive constants $M_1$ and $\lambda_1$ and all $t \geq 0$, where $\mu^{eq}_k = \int_{\Omega}f_k(\mathbf{x})d\mathbf{x}$ for each $k \in \mathcal{V}$ and $\boldsymbol{\mu}^0 \in \mathcal{P}(\mathcal{V})$. In particular, $0$ is a simple eigenvalue of the irreducible operator $\sum_{e \in \mathcal{E}}q_e\mathbf{Q}_e$ and $\boldsymbol{\mu}^{eq}$ is the corresponding unique (up to a scalar multiple) and strictly positive eigenvector. Then $0$ is an eigenvalue for the operator $ \boldsymbol{\mathcal{W}} = -\boldsymbol{\mathcal{M}}_{\mathbf{b}}\boldsymbol{\mathcal{B}}_\mathbf{a} +\sum_e q_ea_{S(e)} \boldsymbol{\mathcal{Q}}_e$ with the corresponding eigenvector $\mathbf{f}$, by construction. We will show that this eigenvalue is simple and is the dominant eigenvalue. Let $\mathbf{g} \in \mathbf{L}^2_\mathbf{a}(\Omega)$ such that $\mathbf{g}$ is not the zero element $\mathbf{0}$ 
and is non-negative a.e. in $\Omega^N$. Defining $\mathbf{h} = (\lambda\mathbb{I} - \boldsymbol{\mathcal{W}})^{-1}\mathbf{g}$ for some $\lambda > 0$, we have that
\begin{align*}
\lambda \int_{\Omega} h_i(\mathbf{x})&d\mathbf{x}  + \int_{\Omega}D_i(B_{a_i}h_i)(\mathbf{x})d\mathbf{x} \nonumber \\ 
-&\sum_{e \in \mathcal{E}} \sum_{j=1}^{N} \int_{\Omega}q_ea_{S(e)}(\mathbf{x}) Q^{ij}_eh_j(\mathbf{x})d\mathbf{x} =  \int_{\Omega }g_i(\mathbf{x})d\mathbf{x}
\end{align*}
for each $i \in \mathcal{V}$. This implies that
\begin{align*}
\lambda \int_{\Omega} h_i(\mathbf{x})d\mathbf{x} ~ - ~ & \sum_{e \in \mathcal{E}} \sum_{j=1}^{N} \int_{\Omega}q_ea_{S(e)}(\mathbf{x}) Q^{ij}_eh_j(\mathbf{x})d\mathbf{x} \nonumber \\
&=\int_{\Omega }g_i(\mathbf{x})d\mathbf{x}
\end{align*}
for each $i \in \mathcal{V}$, which implies that
\begin{equation}
\lambda \int_{\Omega} h_i(\mathbf{x})d\mathbf{x} - \sum_{e \in \mathcal{E}} \sum_{j=1}^{N}  \int_{\Omega}q_e k_{S(e)} Q^{ij}_eh_j(\mathbf{x})d\mathbf{x}=\int_{\Omega }g_i(\mathbf{x})d\mathbf{x}
\end{equation}
for each $i \in \mathcal{V}$ for some positive constants $k_i >0$. The existence of such positive constants is guaranteed, since we assumed that $\mathbf{f}$ is non-negative and hence $\mathbf{h}$ is non-negative. 
However, $ \sum_{e}q_ek_{S(e)} \mathbf{Q}_e$ generates an irreducible semigroup on $\mathbb{R}^N$ whenever $q_e >0 $ implies that $k_e>0$ for all $e \in \mathcal{E}$. Hence, $(\lambda \mathbb{I} - \sum_{e}q_ek_{S(e)} \mathbf{Q}_e)^{-1}$ maps non-negative, nonzero elements of $\mathbb{R}^N$ to strictly positive elements of $\mathbb{R}^N$. This implies that $\int_{\Omega}h_i(\mathbf{x})d\mathbf{x}>0$ for each $i \in \mathcal{V}$. From this, we can conclude that $h_i(\mathbf{x}) > 0$ for a.e. $\mathbf{x} \in \Omega$ for each $i \in \mathcal{V}$. To see this more explicitly, note that $\mathbf{h}$ must satisfy
\begin{equation}
\lambda h_i -D_iB_{a_i}h_i -G^{ii} a_i h_i = g_i + \sum_{j=1,j \neq i}^{N}G^{ij}a_jh_j
\end{equation}
for each $i \in \mathcal{V}$, where $\mathbf{G} = \sum_{e \in \mathcal{E}} q_e\mathbf{Q}_e$. Let $\mathcal{M}_{a_i}$ be the multiplication operator, defined on $L^2(\Omega) = L^2_{a_i}(\Omega)$, that is associated with the function $a_i$. Since $a_i \geq \ell$ for some $\ell>0$, the inverse $R_\lambda^i= (\lambda \mathbb{I} - D_iB_{a_i} -G^{ii}\mathcal{M}_{a_i})^{-1} = (\lambda \mathcal{M}_{a_i}^{-1} - D_iB_{a_i}\mathcal{M}_{a_i}^{-1}  -G^{ii}\mathbb{I})^{-1} \mathcal{M}_{a_i}^{-1}$ exists. The operator $\lambda \mathcal{M}_{a_i}^{-1} - D_iB_{a_i}\mathcal{M}_{a_i}^{-1} $ generates an irreducible semigroup on $L^2(\Omega)$ \cite{ouhabaz2009analysis}[Theorem 4.5] (see equation (4.8) in the cited reference for the class of operators considered); formally, $B_{a_i}\mathcal{M}_{a_i}^{-1}$ is the operator $\nabla \cdot (f_i \nabla (\cdot))$. Hence, $(R_\lambda^i [g_i + \sum_{j=1,j \neq i}^{N}G^{ij}a_jh_j])(\mathbf{x}) $ is strictly positive for a.e. $\mathbf{x} \in \Omega$ and each $i \in \mathcal{V}$, since $\sum_{j=1,j \neq i}^{N}G^{ij}$ and $h_i$ are nonzero for each $i \in \mathcal{V}$. Therefore, $(\lambda \mathbb{I} - \boldsymbol{\mathcal{W}})^{-1} $ maps nonzero, non-negative elements of $\mathbf{L}^2_\mathbf{a}(\Omega)$ to strictly positive elements of $\mathbf{L}^2_\mathbf{a}(\Omega)$. This implies that the semigroup generated by the operator $\boldsymbol{\mathcal{W}}$ is irreducible. Now, we can use \cite{arendt2006one}[Corollary C-III.3.17] to establish that the eigenvalue $0$ is simple and is the dominant eigenvalue.  This follows from the cited corollary because $\boldsymbol{\mathcal{W}}$ has a compact resolvent and generates an analytic semigroup, due to the fact that it is a bounded perturbation of the operator $-\boldsymbol{\mathcal{M}}_{\mathbf{b}}\boldsymbol{\mathcal{B}}_\mathbf{a}$, which itself has a compact resolvent and generates an analytic semigroup \cite{engel2000one}[Proposition III.1.12]. Additionally, we know from \cite{engel2000one}[Corollary III.1.19] that since $\boldsymbol{\mathcal{W}}$ has a compact resolvent, its spectrum is discrete. Then the result follows from \cite{engel2000one}[Corollary V.3.3].
\end{proof}

Irreducibility is not necessary, but only sufficient, for the simplicity of the dominant eigenvalue of a compact positive operator. The goal of the following proposition and theorem is to extend the result in Proposition \ref{eq:stabsc} to a much larger set of equilibrium distributions, for which the resulting semigroup is not necessarily irreducible.

\begin{proposition}
\label{specperb}
Let $\mathbf{P} \in \mathbb{R}^{N \times N}$ be essentially non-negative, i.e., $P^{ij} \geq 0$ for all $i \neq j$ in $\mathcal{\mathcal{V}}$. Let $\boldsymbol{\mathcal{P}}$ be the linear bounded operator on $\mathbf{L}^2(\Omega)$, defined pointwise using $\mathbf{P}$ as $(\boldsymbol{\mathcal{P}}\mathbf{h})(\mathbf{x}) = \mathbf{P}\mathbf{h}(\mathbf{x})$ for a.e. $\mathbf{x} \in \Omega$ for all $\mathbf{h} \in \mathbf{L}^2(\Omega)$. Suppose $\mathbf{b} \in \mathbf{L}^{\infty}(\Omega)$ such that $b_i =D_i\mathbf{1}$ is a positive constant function for each $i \in \mathcal{V}$. In addition, suppose that ${\rm spec}(\mathbf{P})$ lies in $\mathbb{C}_-$. If $a_i = \mathbf{1}$ for each $i \in \mathcal{V}$, then ${\rm spec}(-\boldsymbol{\mathcal{M}}_\mathbf{b} \boldsymbol{\mathcal{B}}_\mathbf{a} +\boldsymbol{\mathcal{P}})$ lies in $\mathbb{C}_-$.

\end{proposition}
\begin{proof}
The proof follows the same line of argument as Theorem \ref{spelc2}. Note that according to the Lie-Trotter product formula, $\boldsymbol{\mathcal{W}} = -\boldsymbol{\mathcal{M}}_\mathbf{b}\boldsymbol{\mathcal{B}}_\mathbf{a} +\boldsymbol{\mathcal{P}}$ generates a positive semigroup since both $-\boldsymbol{\mathcal{B}}_\mathbf{a}$ and $\boldsymbol{\mathcal{P}}$ generate positivity preserving semigroups. Hence, if $\lambda >0 $ is large enough, then $R_\lambda = (\lambda - \boldsymbol{\mathcal{W}})^{-1}$ is a positive operator. Moreover, $R_{\lambda}$ is a compact operator and has a nonzero spectral radius $r$. From the Krein-Rutman theorem \cite{meyer2012banach}[Theorem 4.1.4], it follows that there exists a positive function $\mathbf{h} \in \mathbf{L}^2_{\mathbf{a}}(\Omega)= \mathbf{L}^2(\Omega) $ such that $r\mathbf{h} - R_\lambda \mathbf{h}= \mathbf{0}$.  This implies that $\lambda -\frac{1}{r}$ is an eigenvalue of $\boldsymbol{\mathcal{W}}$. However, this implies that $\int_\Omega -(B_{a_i}h_i) = 0$ for each $i \in \mathcal{V}$, and hence that
\begin{equation}
\left(\lambda -\frac{1}{r}\right)  \int_{\Omega} h_i(\mathbf{x})d\mathbf{x} - \sum_{j=1}^{N} \int_{\Omega}P^{ij}h_j(\mathbf{x})d\mathbf{x} ~=~ 0 \nonumber
\end{equation}
for each $i \in \mathcal{V}$. If $\lambda -\frac{1}{r} \geq 0$, then we arrive at a contradiction, since ${\rm spec}(\mathbf{P})$ lies in $\mathbb{C}_-$. Here, we have used the fact that $\mathbf{h}$ is a positive function, and therefore $\int_{\Omega}h_i(\mathbf{x})d\mathbf{x}$ cannot be equal to $0$ for each $i \in \mathcal{V}$.
\end{proof}

\begin{theorem}
\label{solstbpr}
Let $\mathcal{G} =(\mathcal{V},\mathcal{E})$ be strongly connected, and let $\Omega$ be an extension domain. Let $\mathbf{b} \in \mathbf{L}^{\infty}(\Omega)$ such that $b_i =D_i\mathbf{1}$ is a positive constant function for each $i \in \mathcal{V}$.  Let $\mathbf{f} \in \mathbf{L}^{\infty}(\Omega)$ be such that $f_i \geq c\int_{\Omega }f_i(\mathbf{x})d\mathbf{x} $ for some positive constant $c>0$. Let $\mathcal{V}_1 = \lbrace i \in\mathcal{V}: \int_{\Omega }f_i(\mathbf{x})d\mathbf{x}>0 \rbrace$. Additionally, consider the set $\mathcal{E}_1 = \{e \in \mathcal{E}~:~S(e),T(e) \in \mathcal{V}_1\}$. Suppose that the graph $\mathcal{G}_1=(\mathcal{V}_1,\mathcal{E}_1)$ is strongly connected. Then there exist $\mathbf{a}\in \mathbf{L}^{\infty}(\Omega)$ and spatially-dependent reaction coefficients $\lbrace K_e(\mathbf{x}) \rbrace_{e \in \mathcal{E}} \in \mathbf{L}^{\infty}(\Omega)$ for which $-\boldsymbol{\mathcal{M}}_\mathbf{b}\boldsymbol{\mathcal{B}}_\mathbf{a}+\sum_{e \in \mathcal{E}}K_e\boldsymbol{\mathcal{Q}}_e$ generates a positive semigroup $(\boldsymbol{\mathcal{S}}(t))_{t \geq 0}$ on $\mathbf{L}^2_\mathbf{a}(\Omega)$ such that if $\mathbf{y}^0 \in \mathbf{L}^2_\mathbf{a}(\Omega)$ is a positive function and $\sum_{i \in \mathcal{V}}\int_{\Omega} f_i(\mathbf{x})d\mathbf{x} = \sum_{i \in \mathcal{V}}\int_{\Omega} y_i(\mathbf{x})d\mathbf{x}$, then
\begin{equation}
\|\boldsymbol{\mathcal{S}}(t)\mathbf{y}_0 - \mathbf{f}\| ~\leq~ M e^{-\lambda t}
\end{equation}
for some positive constants $M$ and $\lambda$ and all $t \geq 0$.
\end{theorem}
\begin{proof}
Without loss of generality, we assume that the set $\mathcal{V}_1$ is of the form $\mathcal{V}_1 = \lbrace 1,2, ..., \bar{N} \rbrace$ for some integer $\bar{N} \leq N$. Define $\boldsymbol{\mu}^{eq} \in \mathbb{R}^{N}_{+}$ such that $\mu^{eq}_i = \int_{\Omega}f_i(\mathbf{x})d\mathbf{x}$ for each $i \in \mathcal{V}$. Then from \cite{elamvaz2016lin}[Theorem IV.5] (see \cite{biswal2017mean} for proof), it follows that there exist positive constants $\lbrace q_e \rbrace_{e \in \mathcal{E}}$ such that the solution $\boldsymbol{\mu}(t)$ of the ODE system  
\eqref{eq:ctrsys} converges exponentially to $\boldsymbol{\mu}^{eq}$. In particular, the matrix $\sum_{e \in \mathcal{E}}q_e\mathbf{Q}_e$ has $0$ as a simple eigenvalue with $\boldsymbol{\mu}^{eq}$ as the corresponding eigenvector, which is unique up to a scalar multiple. Let $\mathbf{G} = \sum_{e \in \mathcal{E}}\mathbf{Q}_e$. Then $\mathbf{G}$ is necessarily of the form 
\begin{equation}
~~ \mathbf{G} =  
\begin{bmatrix}
    \mathbf{G}_1       & \mathbf{G}_2  \\
    \mathbf{0} & \mathbf{G}_3
\end{bmatrix},
\end{equation}
where $\mathbf{G}_1 \in \mathbb{R}^{\bar{N} \times \bar{N}}$, $\mathbf{G}_2 \in \mathbb{R}^{\bar{N} \times (N - \bar{N})}$, $\mathbf{G}_3 \in \mathbb{R}^{(N-\bar{N}) \times (N-\bar{N})}$, and $\mathbf{0}$ is the zero element of $\mathbb{R}^{(N - \bar{N}) \times  \bar{N} }$. If $\mathbf{G}$ does not have the block triangular structure above, then there exist indices $i \in \mathcal{V}_1$ and  $j \in \mathcal{V}  \backslash  \mathcal{V}_1$ such that $G^{ji} >0$. But this implies that if $\boldsymbol{\mu}^0 = \boldsymbol{\mu}^{eq}$, then $\dot{\mu}_j(0) \neq  0$ for all $j \in \mathcal{V}$, hence contradicting that $\mathbf{G}\boldsymbol{\mu}^{eq}$ is the zero element of $\mathbb{R}^N$. Moreover, since $\lim_{t \rightarrow \infty}\mu_j(t) =0$ for all $j \in \mathcal{V} \backslash \mathcal{V}_1$ for any $\boldsymbol{\mu}^0 \in \mathbb{R}^{N}$, we must have that ${\rm spec}(\mathbf{G}_3)$ is in $\mathbb{C}_-$ and that $0$ is a simple eigenvalue of $\mathbf{G}_1$. Now, let $\mathbf{a} \in \mathbf{L}^\infty(\Omega)$ be such that $a_i = 1/{f_i} $ if $i \in \mathcal{V}_1 $ and $a_i = k_i \mathbf{1}$ if $i \in \mathcal{V} \backslash \mathcal{V}_1$ for some positive constant $k_i$. Then consider the operator $\boldsymbol{\mathcal{W}}= -\boldsymbol{\mathcal{M}}_\mathbf{b}\boldsymbol{\mathcal{B}}_\mathbf{a} + \sum_{e \in \mathcal{E}}q_e a_S(e)\boldsymbol{\mathcal{Q}}_e$. This operator is of the form  
\begin{equation}
~~\boldsymbol{\mathcal{W}} =  
\begin{bmatrix}
    \boldsymbol{\mathcal{W}}_1       & \boldsymbol{\mathcal{W}}_2  \\
    \mathbf{0} & \boldsymbol{\mathcal{W}}_3
\end{bmatrix},
\end{equation}
where $\boldsymbol{\mathcal{W}}_1 \in \mathcal{L}(X_1,X_1)$, $\boldsymbol{\mathcal{W}}_2 \in \mathcal{L}(X_1, X_2)$, $\boldsymbol{\mathcal{W}}_3 \in \mathcal{L}(X_2,X_2)$, and $\mathbf{0}$ is the zero element of $ \mathcal{L}(X_2,X_1)$, with $X_1 = L^2_{a_1} \times ... \times L^2_{a_{\bar{N}}}$ and $X_2 = L^2_{a_{\bar{N}+1}} \times ... \times L^2_{a_{N}}$. From Proposition \ref{specperb}, it follows that ${\rm spec}(\boldsymbol{\mathcal{W}}_3)$ lies in $\mathbb{C}_-$. Moreover, from Theorem \ref{spelc1}, it follows that $0$ is a simple and dominant eigenvalue of $\boldsymbol{\mathcal{W}}_1$ with the corresponding eigenvector $[f_1~ ...~f_{\bar{N}}]^T$. Then the result follows from \cite{engel2000one}[Corollary V.3.3].
\end{proof}

\section{CONCLUSION} \label{sec:conc}
In this paper, we have proved controllability properties of a system of advection-diffusion-reaction (ADR) PDEs with zero-flux boundary condition that is defined on certain smooth domains. In contrast to previous work, we established controllability of the PDEs with bounded control inputs. Our approach to establishing controllability using spectral properties of the elliptic operators under consideration is also novel. In addition, we have provided constructive solutions to the problem of asymptotically stabilizing a class of hybrid-switching diffusion process (HSDPs) to target non-negative stationary distributions. Future work will focus on extending the arguments in this paper to the case where the corresponding HSDP has diffusion and velocity control parameters in only a small subset of the discrete behavioral states, and the diffusion coefficients are non-constant.

\ifCLASSOPTIONcaptionsoff
  \newpage
\fi



%
\bibliographystyle{plain}
\bibliography{cdcref_v2}
\end{document}